\documentclass[12pt]{article}
\usepackage{natbib} 
\usepackage[latin1]{inputenc}
\usepackage{authblk}    
\usepackage{makeidx}         
\usepackage{graphicx}        
\usepackage{amsmath}
\usepackage{amsthm}
\usepackage{amsfonts}
\usepackage{appendix}
\usepackage{url}
\usepackage{subcaption} 
\usepackage{algorithm}
\usepackage{algpseudocode}
\usepackage{adjustbox} 
\usepackage{bbm}
\newcommand{\bx}{\boldsymbol x}
\newcommand{\bb}{\boldsymbol b}
\newcommand{\bX}{\boldsymbol X}
\newcommand{\bbeta}{\boldsymbol \beta}

\usepackage{array}
\newcolumntype{C}[1]{>{\centering\arraybackslash}m{#1}}

\usepackage{tikz-cd}    
\usetikzlibrary{backgrounds}
\usetikzlibrary{positioning}
\usepackage{tikz}
\usetikzlibrary{arrows,automata} 

\usepackage[colorlinks=true, allcolors=blue]{hyperref}
\usepackage[tikz]{bclogo}
\usetikzlibrary{calc}
\usepackage[font={small}]{caption}

\newtheorem{teo}{Theorem}

\newtheorem{definition}[teo]{Definition}

\theoremstyle{definition}

\newtheorem*{notation*}{Notation and conventions}
\newtheorem{remark}[teo]{Remark}

\begin{document}

 \title{\bf Causal generalized linear models via\\ Pearson risk invariance}
  \author{Alice Polinelli$^a$, Veronica Vinciotti$^a$, Ernst C. Wit$^b$ \\
	\small{$a.$ Department of Mathematics, University of Trento}\\
	\small{$b.$ Institute of Computing, Universit\`a della Svizzera italiana} \\
	}
\date{}
\maketitle
\begin{abstract}
Prediction invariance of causal models under heterogeneous settings has been exploited by a number of recent methods for causal discovery, typically focussing on recovering the causal parents of a target variable of interest.  Existing methods require observational data from a number of sufficiently different environments, which is rarely available. In this paper, we consider a 
structural equation model where the target variable is described by a generalized linear model conditional on its parents. 
Besides having finite moments, no modelling assumptions are made on the conditional distributions of the other variables in the system, and nonlinear effects on the target variable can naturally be accommodated by a generalized additive structure. Under this setting, we characterize the causal model uniquely by means of two key properties: the Pearson risk invariant under the causal model and, conditional on the causal parents, the causal parameters maximize the expected likelihood. These two properties form the basis of a computational strategy for searching the causal model among all possible models. A stepwise greedy search is proposed for systems with a large number of variables. Crucially, for generalized linear models with  a known dispersion parameter, such as Poisson and logistic regression, the causal model can be identified from a single data environment. The method is implemented in the R package \texttt{causalreg}.  
\end{abstract}

\section{Introduction}\label{sec:intro}

Causal inference is the study of determining the causal relationship between two or more variables.  This field has been a subject of research for many years, and its importance has only grown with the rise of artificial intelligence and machine learning \citep{pearl19,luo20}. Indeed, causal inference strategies  go beyond prediction, ensuring that machine learning models have out-of-distribution generalization  and a causal interpretation \citep{arjovsky19}. This is particularly important in applications such as healthcare and autonomous vehicles, where decisions made by machines have real-world consequences \citep{kuang20}. 

In recent years, causal inference has been used extensively in social sciences \citep{sobel00}, economics \citep{varian16}, and public health \citep{glass13} to determine the effectiveness of interventions and to inform policy decisions. For example, randomized controlled trials are commonly used in clinical studies to determine the causal effect of a treatment on a disease. However, conducting a randomized controlled trial is not always feasible or ethical. So, in some cases, observational studies must be used to infer causal relationships. 

One popular method for causal inference in observational studies is propensity score matching \citep{haukoos15}, which involves matching individuals who have similar characteristics and differ only in their exposure to a certain treatment or intervention. This technique has been used in various fields, such as education, finance, and healthcare, to determine the causal effect of a particular intervention. The focus of this propensity score matching is on the estimation of the causal effects, rather than the identification of the causal relationships. The latter is also referred to as causal discovery. An alternative approach to causal inference is the use of structural causal models  \citep{pearl09, bollen13}. These allow researchers also to model the causal relationships between variables, as well as to estimate the strength of these relationships. Within this framework,  a number of methods have been developed for learning the causal graph from observational data, starting from the traditional PC-algorithm \citep{spirtes01} and leading to various alternatives and extensions \citep{glymour19}. In many applications, learning the full causal graph is often too ambitious, and, furthermore, interest is more focused on discovering the direct causes of a certain target variable. This problem, which is also the focus of this paper, is sometimes referred to as causal feature selection \citep{guyon07}.

Embedded in the definition of structural causal models is the concept of causal invariance: the conditional distribution of a variable given its direct causes remains the same under perturbations on any other variable excluding the target variable itself. Causal invariance and the exclusion restriction assumptions are at the basis of invariant causal prediction. The method was recently proposed by \citet{peters16} to identify the causal relationships between variables in a system, and to predict the effect of an intervention on the system. 
The methodology is described primarily for the case of linear structural equation models and requires the availability of data across a number of sufficiently different environments. Searching for invariance across these environments results in conducting a number of statistical tests across all possible causal models. This can become computationally expensive for large systems.
Various extensions have been proposed in recent years to address these limitations. 

\citet{rothenhausler19} avoid the combinatorial search by limiting the source of heterogeneity. They model perturbations as a system of linear structural equations, where covariates are shifted additively by an instrument. In that scenario, the covariance between the covariates and the residuals under the causal model remains invariant across datasets. Given access to two datasets with sufficiently different shifts, the difference between these invariant covariances provides a moment condition that uniquely identifies the causal model. Estimating the moment condition gives rise to the causal Dantzig  estimator. One of the main drawbacks of the causal Dantzig is that the out-of-sample risk under the empirical estimator is often much larger than that of the ordinary least-squares estimator. This is due to the large variance of the causal Dantzig estimator, when the shifts in the different environments are small to moderate. 

Some alternatives to invariant causal prediction focus more closely on out-of-sample prediction guarantees \citep{arjovsky19}.  \citet{rothenhausler21} find that the out-of-sample risk depends on the correlation between the residuals and the instrument. The more uncorrelated they are, the stronger the out-of-sample risk guarantees under unseen covariate shifts, whereby the causal model provides the best guarantees. Assuming access to the instrumental variable that generates the covariate shift, the authors propose to build a regularized estimator, called anchor regression, that regulates the amount of correlation between the residuals and the instrumental variable. While overcoming the issues of causal Dantzig, the main disadvantage of anchor regression is that it requires explicit information about the shift variables. This is often not available in observational studies. 

\citet{kania23} take a different view to anchor regression and aim for a solution that lies between the ordinary least-squares solution, which has the best in-sample risk guarantees, and the causal Dantzig solution, which has the best out-of-sample risk guarantees but which can be attained only with enough heterogeneity in the data. They avoid the need of instrumental variables by adding a regularization term to the likelihood or \emph{pooled risk} that encourages causal invariance. A cross-validation strategy is proposed in order to find the best level of regularization from the available data. 

Causal invariance approaches, and their theoretical guarantees, are mostly presented for the case of linear structural equation models, i.e., linear dependences between the variables and Gaussian distributed errors. A small number of studies have considered extensions to more general models. \citet{buhlmann20} proposes an extension of anchor regression to the nonlinear case, while \citet{kook22} consider a transformation model \citep{hothorn14} for the conditional distribution of the target variable given the parents. Similarly, invariant causal prediction has been extended to the nonlinear case \citep{heinze18} and to the case of time-ordered \citep{pfister19} and general response types \citep{kook23}.  In the latter case, a transformation model is again used for the conditional distribution of the target variable given the parents, while causal invariance is translated into an assumption of invariance of the expected conditional covariance between the environments and the score residuals. As with the linear equation models with Gaussian errors, these extensions also suffer from the requirement of observational data from a number of sufficiently different environments.

This paper proposes a new model-based approach for causal discovery for general response types which, for some important special cases, does not require multiple environments. In particular, we consider a structural equation model where the target variable is described by a generalized linear model conditional on its parents \citep{mccullagh89}. The link function can include flexible, possibly nonlinear, additive modelling structures. As in \citet{kook23}, no modelling assumptions are made on the other conditional distributions. Under this setting, we characterize the causal model uniquely by means of two key properties: invariance of the Pearson risk and achieving the maximum of a conditional population likelihood. As the expected squared Pearson residual of the causal model depends only on the dispersion parameter, causal generalized linear models with a known dispersion parameter, such as Poisson or logistic regression, can be identified from a single environment. Although the results are general, here we will particularly focus on these models. 

In section~\ref{sec:model}, we define the structural equation model formally, including the characterization of the causal model describing the target variable. A statistical approach for detecting the causal model among all possible models from observational data is proposed. A computationally efficient stepwise strategy is presented for systems with many variables. 
Section~\ref{sec:simulation} evaluates the performance of the method on simulated data in the case of Poisson and logistic regression and performs a comparison with the PC-algorithm. Section~\ref{sec:realdata} considers an illustration of the method on data generated under a controlled experiment, before showing its applicability in realistic settings where the causal model is unknown. In particular, we show how the method can be used to identify the causal determinants of women's fertility and high earning, respectively. 

\section{Causal generalized linear model} \label{sec:model}

In this section, we define the proposed method. We begin by defining a structural equation model with exponential family target. Then we define a characterization of the causal model, that is unique under certain conditions. Given this characterization, we will define a population and empirical algorithm for the search of the causal model from observational data. 

\subsection{Structural equation model with exponential family target}
We define a structural equation model that describes the dependences between a set of random variables. Let $Y$ be the target variable of interest and let $\bX=(X_1,\ldots,X_p)^\top$ be the remaining variables in the system. The aim is to discover which of these variables are the causal parents of $Y$ and what their effect on $Y$ is. The following definition is considered for the joint distribution of $(\bX,Y)$.
\begin{definition} \label{def:model}
A structural equation model $(\bX,Y)$ with exponential dispersion family target $Y$ satisfies the following two conditions:
\begin{enumerate}
\item $(\bX, Y)$ is a structural causal model \citep{pearl09}; 
\item $Y | \bX_{\mbox{\scriptsize PA}} = \bx_{\mbox{\scriptsize PA}} \sim EDF\Big(b(f_{\mbox{\scriptsize PA}}(\bx_{\mbox{\scriptsize PA}})),a(\phi)\Big)$ with $\bX_{\mbox{\scriptsize PA}}$ the parents of $Y$ in the causal graph, where
\[EDF(b(\theta),a(\phi))=\exp\bigg\{\frac{y\theta-b(\theta)}{a(\phi)}+c(y;\phi)\bigg\}, \quad \quad \theta=f_{\mbox{\scriptsize PA}}(\bx_{\mbox{\scriptsize PA}}),\]
and $f_{\mbox{\scriptsize PA}}: \mathbb{R}^{|PA|} \rightarrow \mathbb{R}$ depends on $\bx\in\mathbb{R}^p$ only through $\bx_{\mbox{\scriptsize PA}}$, while $a(\cdot)$, $b(\cdot)$ and $c(\cdot)$ are known functions that define the distribution in the exponential dispersion family \citep{mccullagh89}.
\end{enumerate}
Additional regularity conditions on $\bX$ and $f_{\mbox{\scriptsize PA}}$, such as finite moments for $f_{\mbox{\scriptsize PA}}(\bX)$ in the case of Poisson distributed $Y$, as well as a strictly positive joint density $p(\bX)>0$, should be imposed depending on the the type of exponential family selected. We furthermore assume that the graph is faithful with respect to $Y$.  
\end{definition}

The distribution of $Y$ conditional on its parents is described by a generalized linear model with a canonical link, while no modelling assumptions are made on the distributions of the predictors $X_i$ given their parents. Notable examples in this family, frequently used on a variety of application studies, are Poisson and logistic regression. In both of these cases, $a(\phi)$ is known and equal to 1, while $b(\theta) = e^{\theta}$ for Poisson and $b(\theta) = \log(1+e^{\theta})$ for the logistic regression model. In the following, we will consider closely these distributions which belong to the exponential dispersion family and are defined by only one parameter. Besides Poisson and Binomial, other examples include the exponential and logarithmic distributions. Note that in our definition of the model, linearity is not assumed in the distribution of $Y$ given $\bX_{\mbox{\scriptsize PA}}$.  The special case of $f_{\mbox{\scriptsize PA}}(\bx_{\mbox{\scriptsize PA}})=\bbeta_{\mbox{\scriptsize PA}}^\top\bx_{\mbox{\scriptsize PA}}$ for some causal parameters $\bbeta_{\mbox{\scriptsize PA}}$ may be of interest in practice, but the theoretical derivations do not assume this and apply also to more general models like generalized additive models \citep{wood17}.

\subsection{Characterization of causal model} \label{sec:characterization}
The structural equation model defined above contains a causal parameter, such as the infinite dimensional functional parameter $f_{\mbox{\scriptsize PA}}$ or the finite dimensional linear parameter $\bbeta_{\mbox{\scriptsize PA}}$. The identification of this parameter identifies the causal structure governing $Y$. The next theorem describes the properties that characterise the causal model, that is the properties that are satisfied by $f_{\mbox{\scriptsize PA}}$.
\begin{teo} \label{th:causal}
Let $(\bX,Y)$ be distributed according to a structural equation model defined in definition~\ref{def:model}. Let $\bX_{\mbox{\scriptsize PA}}$ be the parents of $Y$ in this graph and $f_{\mbox{\scriptsize PA}}$ the true causal model linking $Y$ to its parents. Let $\dot{b}(\theta)$ and $\ddot{b}(\theta)$ be the first and second derivative of the cumulant generator function, respectively. 

Then, $f_{\mbox{\scriptsize PA}}$ satisfies the following population conditions:
	\begin{enumerate}
		\item  $f_{\mbox{\scriptsize PA}}$ is the maximizer of the expected likelihood of $Y$ and its causal parents,
		\[ f_{\mbox{\scriptsize PA}} = \arg\max_{f} \mathbb{E}_{\bX_{\mbox{\scriptsize PA}},Y}\left[\ell_{\bX_{\mbox{\scriptsize PA}},Y}(f) \right];  \]
		\item $f_{\mbox{\scriptsize PA}}$ achieves a perfectly dispersed Pearson risk,
		\[   \mathbb{E}_{\bX,Y}  \left[\frac{(Y- \dot{b}(f_{\mbox{\scriptsize PA}}(\bX)))^2}{\ddot{b}(f_{\mbox{\scriptsize PA}}(\bX))}\right] = a(\phi),   \]
	\end{enumerate}
for any arbitrary square-integrable distribution of $\bX$. In particular, this may be the observational or any interventional distribution. 
\end{teo} 
\begin{proof}
The first property is a direct consequence of the maximum likelihood estimator maximizing the expected log-likelihood, i.e., if we restrict the system to only $Y$ and its parents, then 
\[\mathbb{E}_{\bX_{\mbox{\scriptsize PA}},Y}\left[\ell_{\bX_{\mbox{\scriptsize PA}},Y}(f) \right] = \mathbb{E}_{\bX_{\mbox{\scriptsize PA}},Y}\bigg[\frac{Yf(\bX_{\mbox{\scriptsize PA}})-b(f(\bX_{\mbox{\scriptsize PA}}))}{a(\phi)}+c(Y;\phi) \bigg]\]
is maximized by $f_{\mbox{\scriptsize PA}}$.

As for the second property, using the laws of conditional probabilities, we can write the Pearson risk as
\[\mathbb{E}_{\bX,Y}\bigg[\frac{(Y-\dot{b}(f_{\mbox{\scriptsize PA}}(\bX)))^2}{\ddot{b}(f_{\mbox{\scriptsize PA}}(\bX))}\bigg]=\mathbb{E}_{\bX_{\mbox{\scriptsize PA}}}\bigg[\mathbb{E}_{Y,\bX_{{\mbox{\scriptsize PA}}^C}}\bigg[\frac{(Y-\dot{b}(f_{\mbox{\scriptsize PA}}(\bX)))^2}{\ddot{b}(f_{\mbox{\scriptsize PA}}(\bX))}\bigg|\bX_{\mbox{\scriptsize PA}}\bigg]\bigg]\]
where ${\mbox{\scriptsize PA}}^C$ is the complementary set of the parental set of $Y$, and hence $\bX_{{\mbox{\scriptsize PA}}^C}$ are the non-causal variables.
Recalling that the expected values in the likelihood and Pearson risks are taken w.r.t. $(\bX,Y)$ distributed according to the true structural causal model, we have in particular that $Y|\bX_{\mbox{\scriptsize PA}}=\bx_{\mbox{\scriptsize PA}}\sim EDF\Big(b(f_{\mbox{\scriptsize PA}}(\bx_{\mbox{\scriptsize PA}})),a(\phi)\Big)$ for the chosen EDF distribution and $f_{\mbox{\scriptsize PA}}(\bX)$ is only a function of $\bX_{\mbox{\scriptsize PA}}$. Moreover, being a generalized linear model, it holds that $\mathbb{E}[Y|\bX_{\mbox{\scriptsize PA}}]=\dot{b}(f_{\mbox{\scriptsize PA}}(\bX_{\mbox{\scriptsize PA}}))$ and $\mathbb{V}ar[Y|\bX_{\mbox{\scriptsize PA}}]=\ddot{b}(f_{\mbox{\scriptsize PA}}(\bX_{\mbox{\scriptsize PA}}))a(\phi)$. Using these properties on the inner expectation, we can write the Pearson risk as
\[\mathbb{E}_{\bX,Y}\bigg[\frac{(Y-\dot{b}(f_{\mbox{\scriptsize PA}}(\bX)))^2}{\ddot{b}(f_{\mbox{\scriptsize PA}}(\bX))}\bigg]=\mathbb{E}_{\bX_{\mbox{\scriptsize PA}}}\bigg[\frac{\ddot{b}(f_{\mbox{\scriptsize PA}}(\bX_{\mbox{\scriptsize PA}}))a(\phi)}{\ddot{b}(f_{\mbox{\scriptsize PA}}(\bX_{\mbox{\scriptsize PA}}))}\bigg]=a(\phi).\]
\end{proof}

The second property in Theorem \ref{th:causal}  shows how the Pearson risk is invariant under the causal model, that is it is stable under any distribution of the covariates $\bX$, including possible changes due to interventions. Unlike the quadratic risk used in the Gaussian case \citep{kania23, rothenhausler19}, it is clear how an invariant property for generalized linear models must account also for their inherent heteroscedasticity. This is the role played by the conditional variance at the denominator of the Pearson risk. However, this invariance property can be satisfied by a number of other models.  Considering for simplicity a generalized linear model with a linear predictor described by a $p+1$-dimensional vector of parameters $\bbeta$, the set of solutions for which the Pearson risk condition is satisfied corresponds to a $p$ dimensional manifold in $\mathbb{R}^{p+1}$. The first condition in Theorem \ref{th:causal}, stating that the causal model maximizes the expected likelihood of the marginal model $(\bX_{\mbox{\scriptsize PA}},Y)$, guarantees an almost sure unique identification of the causal model within this surface, under the conditions described by the next theorem. 

\begin{teo} \label{th:uniqueness}
Let $S \subset \{1,\ldots,p\}$ indicate a potential set of causal parents of $Y$ in $\bX$. Let $f_S: \mathbb{R}^{|S|} \rightarrow \mathbb{R}$ be the function that depends strictly on $\bX_S$, and that satisfies (i) $f_S$ is defined on some functional space $\cal{F}$ with basis $\psi_S=\{\psi_1,\psi_2,\ldots\}$, i.e., $f_S(\bx_S)=\bbeta_S^\top\psi_S(\bx_S)$, and (ii) $\bbeta_S$ solves the expected likelihood score equations
\[\mathbb{E}_{\bX_S,Y}\left[ \psi_S(\bX_S)Y\right] - \mathbb{E}_{\bX_S,Y}\left[\psi_S(\bX_S)\dot{b}(\bbeta^\top\psi_S(\bX_S)) \right] = 0.\]
 Then 
\[S=PA \mbox{ and } f_S=f_{\mbox{\scriptsize PA}} \mbox{ a.s.} \quad \mbox{iff} \quad   \mathbb{E}_{\bX,Y}  \left[\frac{(Y- \dot{b}(\bbeta_S^\top\psi_S(\bX_S)))^2}{\ddot{b}(\bbeta_S^\top\psi_S(\bX_S))}\right] = a(\phi).\]
\end{teo}
\begin{proof}
If $S=PA$ and $f_S=f_{\mbox{\scriptsize PA}}$ a.s., then Theorem~\ref{th:causal} shows how $f_{\mbox{\scriptsize PA}}=\bbeta_{\mbox{\scriptsize PA}}^\top\psi_{\mbox{\scriptsize PA}}$ satisfies the likelihood score equations from the first property and the Pearson risk invariance from the second property. We therefore only need to prove the reverse direction of this theorem.\\
Consider then a $\bbeta_S$ that solves the score equations and that satisfies Pearson risk invariance. Then, since $\bbeta_S^\top\psi_S$ depends only on $\bX_S$, using similar derivations to before, we can rewrite the Pearson risk as
\[\mathbb{E}_{\bX,Y}\bigg[\frac{(Y-\dot{b}(\bbeta_S^\top\psi_S(\bX_S)))^2}{\ddot{b}(\bbeta_S^\top\psi_S(\bX_S))}\bigg]=\mathbb{E}_{\bX_{S}}\bigg[\frac{\mathbb{E}_{Y}[(Y-\dot{b}(\bbeta_S^\top\psi_S(\bX_S)))^2|\bX_{S}]}{\ddot{b}(\bbeta_S^\top\psi_S(\bX_S))}\bigg].\]
Thus, except for a null set on the space of distributions on $\bX_S$, Pearson risk invariance can be achieved only if 
\[\mathbb{E}_{Y}[(Y-\dot{b}(\bbeta_S^\top\psi_S(\bX_S)))^2|\bX_{S}] = \ddot{b}(\bbeta_S^\top\psi_S(\bX_S))a(\phi).\]
Multiplying both sides of the equation by $\psi_S(\bX_S)$ and dividing by $a(\phi)^2$, we have
\[\frac{\mathbb{E}_{Y}[\psi_S(\bX_S)^\top(Y-\dot{b}(\bbeta_S^\top\psi_S(\bX_S)))^2\psi_S(\bX_S)|\bX_{S}]}{a(\phi)^2}= \frac{\psi_S(\bX_S)^\top\ddot{b}(\bbeta_S^\top\psi_S(\bX_S))\psi_S(\bX_S)}{a(\phi)}.\]
Since $Y-\dot{b}(\bbeta_S^\top\psi_S(\bX_S))$ is a scalar, we can further rewrite the left hand side as an inner product, resulting in
\begin{align*}
&\mathbb{E}_{Y}\bigg[\frac{(\psi_S(\bX_S)(Y-\dot{b}(\bbeta_S^\top\psi_S(\bX_S)))^\top}{a(\phi)}\frac{\psi_S(\bX_S)(Y-\dot{b}(\bbeta_S^\top\psi_S(\bX_S)))}{a(\phi)}|\bX_{S}\bigg]\\
&= \frac{\psi_S(\bX_S)^\top\ddot{b}(\bbeta_S^\top\psi_S(\bX_S))\psi_S(\bX_S)}{a(\phi)}.
\end{align*}
Taking now expectation w.r.t. $\bX_S$ we have that
\begin{align*}
&\mathbb{E}_{\bX_S}\bigg[\mathbb{E}_{Y}\bigg[\frac{(\psi_S(\bX)(Y-\dot{b}(\bbeta_S^\top\psi_S(\bX_S))))^\top}{a(\phi)}\frac{\psi_S(\bX_S)(Y-\dot{b}(\bbeta_S^\top\psi_S(\bX_S)))}{a(\phi)}|\bX_{S}\bigg]\bigg] \\
&= \mathbb{E}_{\bX_S}\bigg[\psi_S(\bX_S)^\top\frac{\ddot{b}(\bbeta_S^\top\psi_S(\bX_S))}{a(\phi)}\psi_S(\bX_S)\bigg].
\end{align*}
We finally notice how the LHS is the expected value of the expected square of the gradient of the maximum log-likelihood of an EDF distribution, since $\bbeta_S$ solves the score equations by assumption. The right hand side is the expected value of the expected Fisher information matrix on $\bbeta_S$ of the same EDF distribution. The two sides match if and only if
\[Y \: | \: \bX_S \sim EDF(b(\bbeta_S^\top\psi_S(\bX_S)),a(\phi)).\]
Except for a null set on the space of distributions on $\bX_S$, under the true structural  causal model this is true only if $S=PA$ and $f_S=f_{\mbox{\scriptsize PA}}$.
\end{proof}

\begin{remark}
Uniqueness is achieved under the restriction that $f_S$ depends on \emph{all} $\bx_S$ --- in Theorem~\ref{th:uniqueness} this was referred to as \emph{strict} dependence on $\bx$. Any set $\bX_S$ that contains $\bX_{\mbox{\scriptsize PA}}$ together with variables that are d-separated from $Y$ by its parents would still satisfy the two properties that characterize the causal model. This is because $f_S$ is simply constant with respect to these variables. 
\end{remark}
In the next section, we propose a computational strategy for the search of the causal model based on the characterization that we have provided above and addressing the redundancy in the candidate models that we have just remarked.

\begin{algorithm}[tb!]
	\begin{enumerate}
		\item For any set of covariates $\bX_{S}$, with $S \subset \{1,\ldots,p\}$, determine the maximum of the expected likelihood, 
		\[ \bbeta_{S} = \arg\max_{\bbeta} \mathbb{E}_{\bX_S,Y}[\ell_{\bX_{S},Y}(\bbeta)].\]
		\item Check whether the Pearson risk is perfectly dispersed for $\bbeta_{S}$, that is whether
		\[\mathbb{E}_{\bX,Y}  \left[\frac{(Y- \dot{b}(\bbeta_S^\top\psi_S(\bX_S)))^2}{\ddot{b}(\bbeta_S^\top\psi_S(\bX_S))}\right] = a(\phi).	\]
		If it is, then $\bbeta_{S}^\top\psi_S$ is a candidate for $f_{\mbox{\scriptsize PA}}=\bbeta_{\mbox{\scriptsize PA}}^\top\psi_{\mbox{\scriptsize PA}}$.
		\item Among the set of candidate models $\bbeta_{S}$, select $\bbeta_{\mbox{\scriptsize PA}}$ as the one minimizing the Bayesian Information Criterion (BIC).
	\end{enumerate}
	\caption{Population algorithm for finding causal parents. \label{alg:pop}}
\end{algorithm}

\subsection{Population and empirical algorithms}
Theorem~\ref{th:uniqueness} can be translated into a population algorithm for finding the causal model. This involves a combinatorial search among all subsets of covariates, and is described in Algorithm~\ref{alg:pop}. The algorithm involves a search of the models that achieve a perfectly dispersed Pearson risk. Then, an information criterion is used to select the optimal model among the class of candidate models. The choice of an information criterion is motivated by the fact that variables that are d-separated from $Y$ by its causal parents, such as ancestors of $Y$,  have no predictive power on $Y$. In principle, any information criterion that penalizes the expected likelihood by a term including the number of covariates can be used at this stage.

\begin{algorithm}[tb!]
	\begin{enumerate}
		\item For any set of covariates $\bX_{S}$, with $S \subset \{1,\ldots,p\}$, find the penalized maximum likelihood estimator, 
		\[ \widehat{\bbeta}_{S} = \arg\max_{\bbeta} \sum_{i=1}^n \frac{y_i \bbeta^\top\tilde\psi_S({\bx_{i}}_{S})-b(\bbeta^\top\tilde\psi_S({\bx_{i}}_{S}))}{a(\phi)} + \lambda P_\lambda(\bbeta),\]
		for a finite dimensional basis $\tilde\psi_S$.
 		\item Find $S_1,\ldots, S_K \subset \{1,\ldots, p \}$ such that for $k=1,\ldots,K$
		\[H_0: \mathbb{E}_{\bX,Y}  \left[\frac{(Y- \dot{b}({\bbeta_S}_{k}^\top\psi_S(\bX_S)))^2}{\ddot{b}({\bbeta_S}_{k}^\top\psi_S(\bX_S))}\right] = a(\phi)\]
	  cannot be rejected
for some significance level $\alpha$. 
\item Among the set of candidate models $f_{S_1}, \ldots, f_{S_K}$, select $f_{\mbox{\scriptsize PA}}$ as the one minimising the Bayesian Information Criterion (BIC).
	\end{enumerate}
	\caption{Full empirical algorithm for finding causal parents. \label{alg:emp}}
\end{algorithm} 
For a finite sample, $(\bx_i, y_i)$, $i=1,\ldots,n$, an empirical version of the population algorithm is proposed in Algorithm~\ref{alg:emp}. This involves replacing the expected likelihood by the empirical likelihood, while the requirement of a perfectly dispersed Pearson risk is checked via an appropriate statistical test. In particular, for each candidate subset $\bX_{S}$, we test the null hypothesis 
\[H_0: \mathbb{E}_{\bX,Y}  \left[\frac{(Y- \dot{b}({\bbeta_S}^\top\psi_S(\bX_S)))^2}{\ddot{b}({\bbeta_S}^\top\psi_S(\bX_S))}\right] = a(\phi).\]
We consider a $q$-dimensional basis $\tilde\psi_S$ for $\{ f:\mathbb{R}^{|S|}\rightarrow\mathbb{R}~|~f= \beta^\top \tilde\psi_S\}$. Typical examples are (i) a simple linear basis, in which case $q=|S|$, or (ii) a simple additive model with individual $q_0$-dimensional bases for each of the variables, in which case $q=q_0|S|$.
For this, we use the test statistic
\[R^{+}_S=\displaystyle\sum_{i=1}^n \left[\frac{(y_i- \dot{b}({\widehat\bbeta}_{S}^\top\tilde\psi_S(\bx_{i,S})))^2}{\ddot{b}({{\widehat \bbeta}_{S}}^\top\tilde\psi_S(\bx_{i,S}))a(\phi)}\right],\]
where $\bx_{i,S}$ denotes the $i$-th realization of $\bX_S$ and $\widehat{\bbeta}_{S}$ is a penalized maximum likelihood estimator,
\[ \widehat{\bbeta}_{S} = \arg\max_{\bbeta} \sum_{i=1}^n \frac{y_i \bbeta^\top\tilde\psi_S(\bx_{i,S})-b(\bbeta^\top\tilde\psi_S(\bx_{i,S}))}{a(\phi)} + \lambda P_\lambda(\bbeta),\]
for some convenient smoothness or sparsity penalty $P_\lambda$ \citep{wood17}. 

The distribution of the test statistic under the null hypothesis is generally unknown, so bootstrap can be used to assess the evidence against the null hypothesis. The algorithm considers all possible models and retains those for which the null hypothesis cannot be rejected. Since the aim is to select perfectly dispersed models, a two-sided test is needed. As with the population algorithm, the candidate set of models obtained via statistical testing is finally refined via BIC, in order to find the sparse most predictive model among the class of perfectly dispersed models.

For some generalized linear models, an approximate distribution of the test statistic under the null hypothesis is available, leading to computational efficiency. Of special notice is the case of Poisson regression: under the null hypothesis and some standard conditions on the penalty, the Pearson statistic can be shown to be asymptotically chi-squared distributed. In particular, it holds, approximately, for the sum of the squared Pearson residuals that
\begin{equation}
	R^{+}_S \overset{H_0}{\sim} \chi^2_{n-\hat{q}(S)},
	\label{eq:test}
\end{equation}
with $\hat{q}(S)$ an estimate for the effective degrees of freedom of the model associate with variables $S$. Therefore, in this case, the acceptance region for the model that includes variables $S_k$ can be efficiently approximated by  
\[\chi^2_{n-\hat{q}_k,\frac{\alpha}{2}} \le \displaystyle\sum_{i=1}^n \left[\frac{(y_i- \dot{b}({\widehat\bbeta}_{S_k}^\top\tilde\psi_{S_k}(\bx_{i,S_k})))^2}{\ddot{b}({{\widehat \bbeta}_{S_k}}^\top\tilde\psi_{S_k}(\bx_{i,S_k}))a(\phi)}\right] \le \chi^2_{n-\hat{q}_k,1-\frac{\alpha}{2}},\]
for some significance level $\alpha$, where $\hat{q}_k$ is the effective degrees of freedom associated with the model involving $S_k$. 

\subsection{Stepwise algorithm}

The full algorithm considers all possible, i.e., $2^p$, models, which leads to an exponential computational complexity of the algorithm. For applications involving a large number of covariates, this may be prohibitively slow. We therefore propose also a stepwise version of the algorithm. The algorithm is is described in detail in Algorithm~\ref{alg:step} and is split into two phases. Starting from an empty model, in the first phase we add variables one-by-one until we get the largest model that we cannot reject to be perfectly dispersed. In a second phase, we remove superfluous, hopefully d-separated, variables by means of a BIC procedure. 

\begin{algorithm}[tb!]
	\begin{enumerate}
		\item Select a threshold $\alpha$. 
		\item Start with an empty model involving just the intercept; 
		\item Calculate the p-value $p$ associated with the test of the perfect dispersion of the squared Pearson residuals. 
		\item Iterate
		\begin{enumerate}
			\item For all available non-included variables $X_j$, fit the model with that variable included to the current model.
			\item Calculate p-value $p_j$ for the test of perfect dispersion associated with addition of variable $X_j$ to the existing model.
			\begin{itemize}
				\item Add variable $X_k$ to the existing model, if 
				$p_k = \max_j p_j$ and $p_k>p$ or $p_k> \alpha$. Continue with the iteration, setting $p = p_k$. 
				\item Stop with the iteration if $\max p_j < p$ and $\max p_j < \alpha$, and pass current model to step 5. 
			\end{itemize}  
		\end{enumerate}
		\item Starting from the model selected in step 4 with BIC $b$, iterate:
		\begin{enumerate}
			\item For all variables $X_j$ included in the current model, fit the model $M_{-j}$ without that variable $X_j$. 
			\item Calculate BIC $b_j$ for model $M_{-j}$ without variable $X_j$.
			\begin{itemize}
				\item Remove variable $X_k$ from the existing model, if $b_k = \min_j b_j$ and $b_k<b$. Continue with the iteration, setting $b=b_k$.
				\item Stop with the iteration, if $\min_j b_j>b$ and return the current model.  
			\end{itemize}
			\end{enumerate}
	\end{enumerate}
	\caption{Stepwise algorithm for finding causal parents. \label{alg:step}}
\end{algorithm}

\section{Simulation study} \label{sec:simulation}

In this section, we illustrate the performance of the procedure. 
In section~\ref{sec:simulation-pop} we explain the idea behind Theorem~\ref{th:causal} as the characterization of the causal parameters in the population scenario by means of a large sample.  In the following two sections, we show the performance of the method in two important cases where causal discovery can be achieved from a single observational environment. In particular, in section~\ref{sec:simulation-finite} we consider a Poisson regression setting for the target of interest, while in section~\ref{sec:simulation-logistic} we consider the case of logistic regression.

\subsection{Proof-of-concept: population scenario}
\label{sec:simulation-pop}
In this section, we highlight the main features of the proposed method via a small example. We consider the DAG in Figure~\ref{fig:sim2covariatesY}, with a target variable $Y$ having a causal parent $X_1$ and a child $X_2$.  
\begin{figure}[tb!]
    \centering
    \begin{subfigure}{.49\linewidth}
    \centering
\begin{tikzpicture}
	\node[state] (X1)  {$X_1$};
    \node[state] (Y) [right = 1cm of X1,blue]  {$Y$};
	\node[state] (X2)[right = 1cm of Y] {$X_2$};	
	\path[draw,thick,->]
	(X1) edge  (Y)
	(X1) edge[bend right] (X2)
	(Y) edge (X2);
\end{tikzpicture}
\caption{}
\end{subfigure}
\begin{subfigure}{.49\linewidth}
\begin{equation*}
 \begin{cases}
  X_1&=\epsilon_1 \\
  Z &= X_1+\epsilon_Z,\\ 
	Y & =F_{\rm{PA}}^{-1}(\Phi_{0,\sigma^2_Z}(\epsilon_{Z}))\\
  X_2&=Z+\epsilon_2, 
\end{cases}
\end{equation*}
\caption{\label{fig:sim2covariatesSEM}}
\end{subfigure}
\caption{Structural causal model for generating data on a target variable $Y$, $p=2$ predictors and the causal graph (a). The structural equations in (b) via a latent variable $Z$ show how the target $Y$ has a $\mbox{Poisson}(e^{X_1})$ distribution $F_{\rm{PA}}$ conditional on the causal parent $X_1$.}
\label{fig:sim2covariatesY}
\end{figure}
The formulation of the structural equations via a latent variable $Z$ (Figure \ref{fig:sim2covariatesSEM}) allows to control more easily the predictive power of the children of $Y$, since this variable is in the same scale as the predictors $\bX$. In this way, we can make the detection of the causal parameters $\bbeta_{\mbox{\scriptsize PA}}$ sufficiently challenging. For the simulation, we assume the noise components $\epsilon_1$, $\epsilon_2$ and $\epsilon_Z$ to be independent and identically distributed standard normal random variables, while the target $Y$ has a Poisson distribution $F_{\rm{PA}}$ conditional on its causal parent $X_1$, given by
\[Y | \bX_{\mbox{\scriptsize PA}} = \bx_{\mbox{\scriptsize PA}} \sim \mbox{Poisson}(e^{x_1}).\]

We simulate a dataset from this system with a large sample size ($n=100000$), i.e., mimicking the population setting. Figure~\ref{fig:contour-likelihood-insample} shows the contour lines of the likelihood function of a full model, with both $X_1$ and $X_2$ being potential parents to $Y$. Since $Y$ has a conditional Poisson distribution and considering only the terms dependent on $\bbeta$, this is given by
\[\ell(\bbeta)=\frac{1}{n}\sum_{i=1}^n [y_i \bbeta^\top\bx_{i}-e^{\bbeta^\top\bx_{i}}].\]
For easy of visualization, we consider a full model without the intercept, and so a $\bbeta$ vector of length 2.
As expected, Figure \ref{fig:contour-likelihood-insample} shows how the Maximum Likelihood (ML) estimate of $\bbeta$ achieves the maximum of this function, being the most predictive model on the observed data. On the other hand, the causal parameter $(1,0)$ lies on a lower contour line. 
\begin{figure}[h!]
\centering
\begin{subfigure}{.49\linewidth}
\centering
\includegraphics[width=\textwidth]{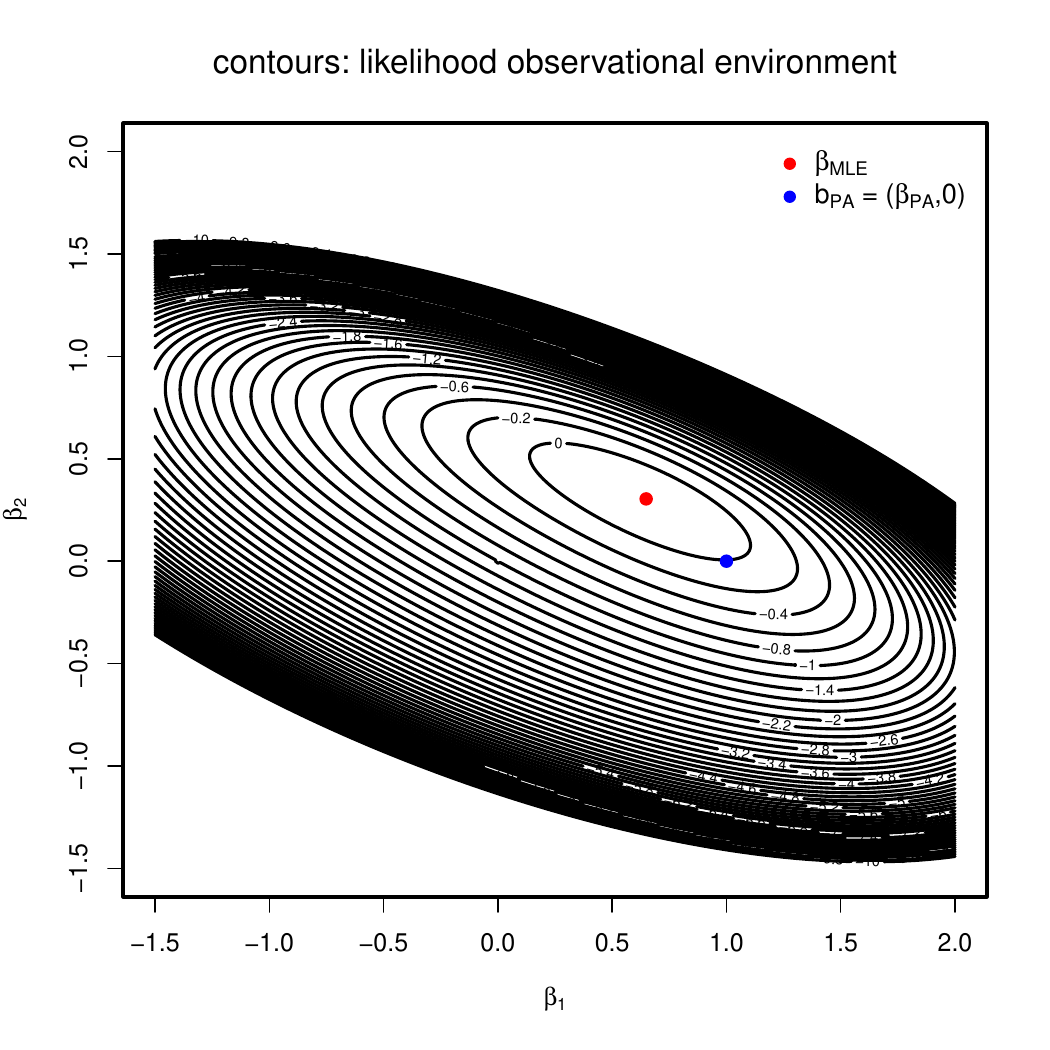}
\caption{\label{fig:contour-likelihood-insample}}
\end{subfigure}
\begin{subfigure}{.49\linewidth}
\centering
\includegraphics[width=\textwidth]{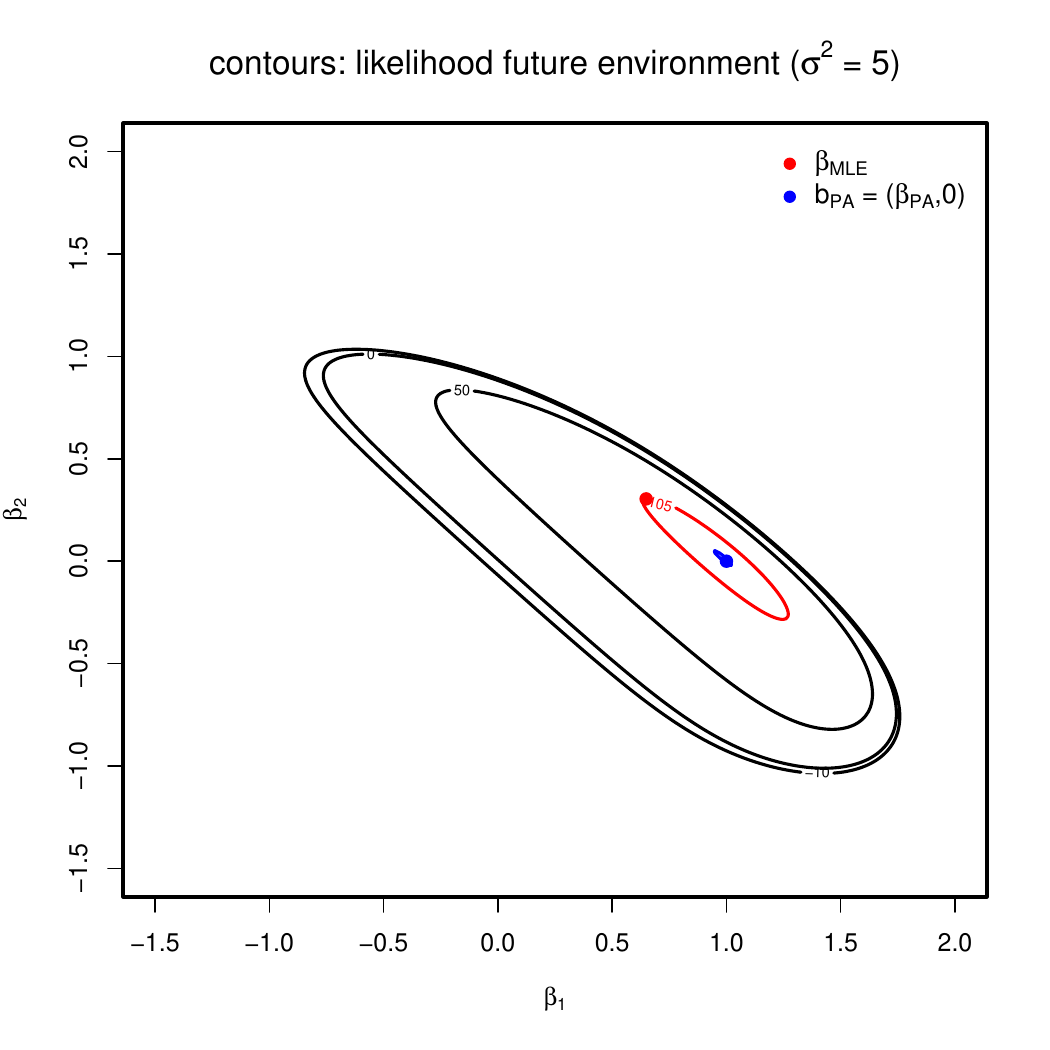}
\caption{\label{fig:contour-likelihood-fut5}}
\end{subfigure}
\begin{subfigure}{.49\linewidth}
\centering
\includegraphics[width=\textwidth]{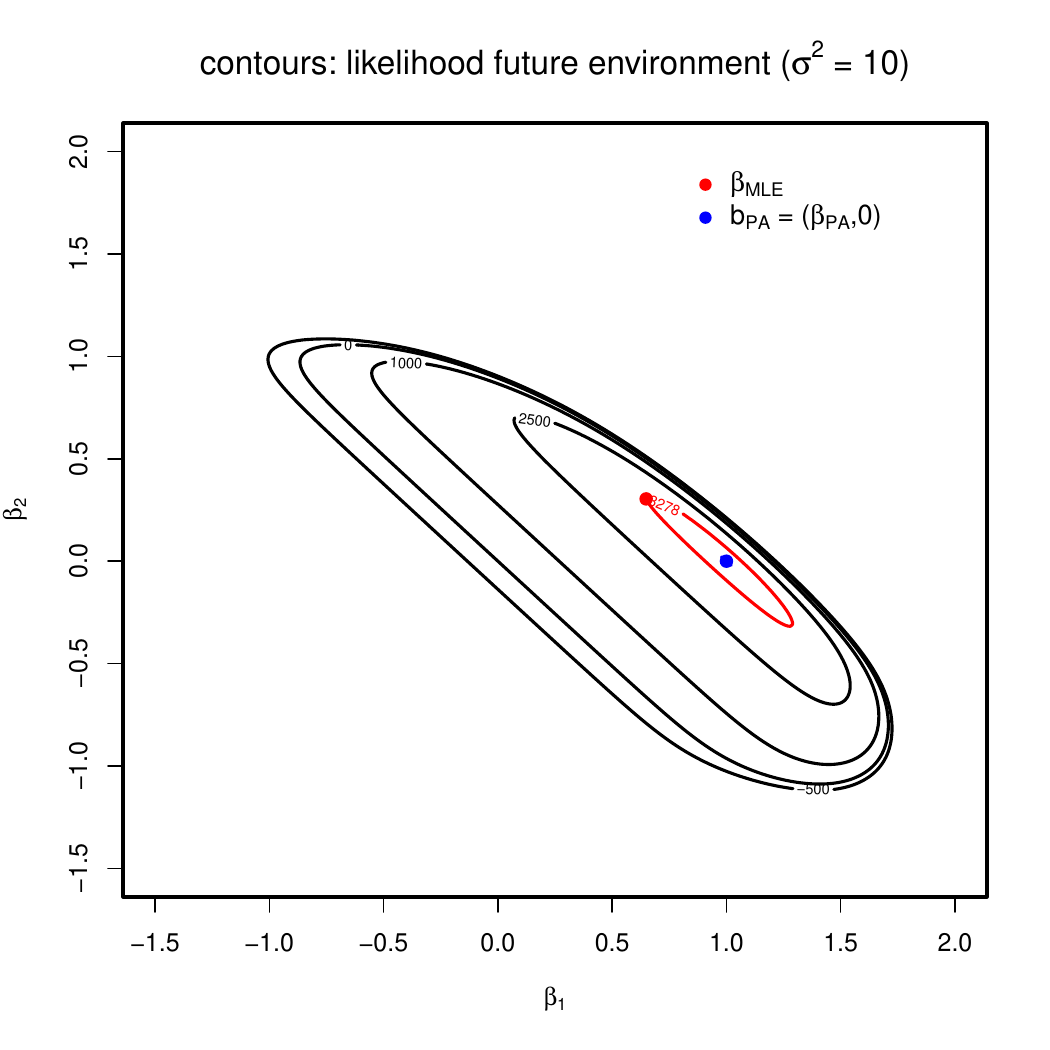}
\caption{\label{fig:contour-likelihood-fut10}}
\end{subfigure}
\begin{subfigure}{.49\linewidth}
	\centering
	\includegraphics[width=\textwidth]{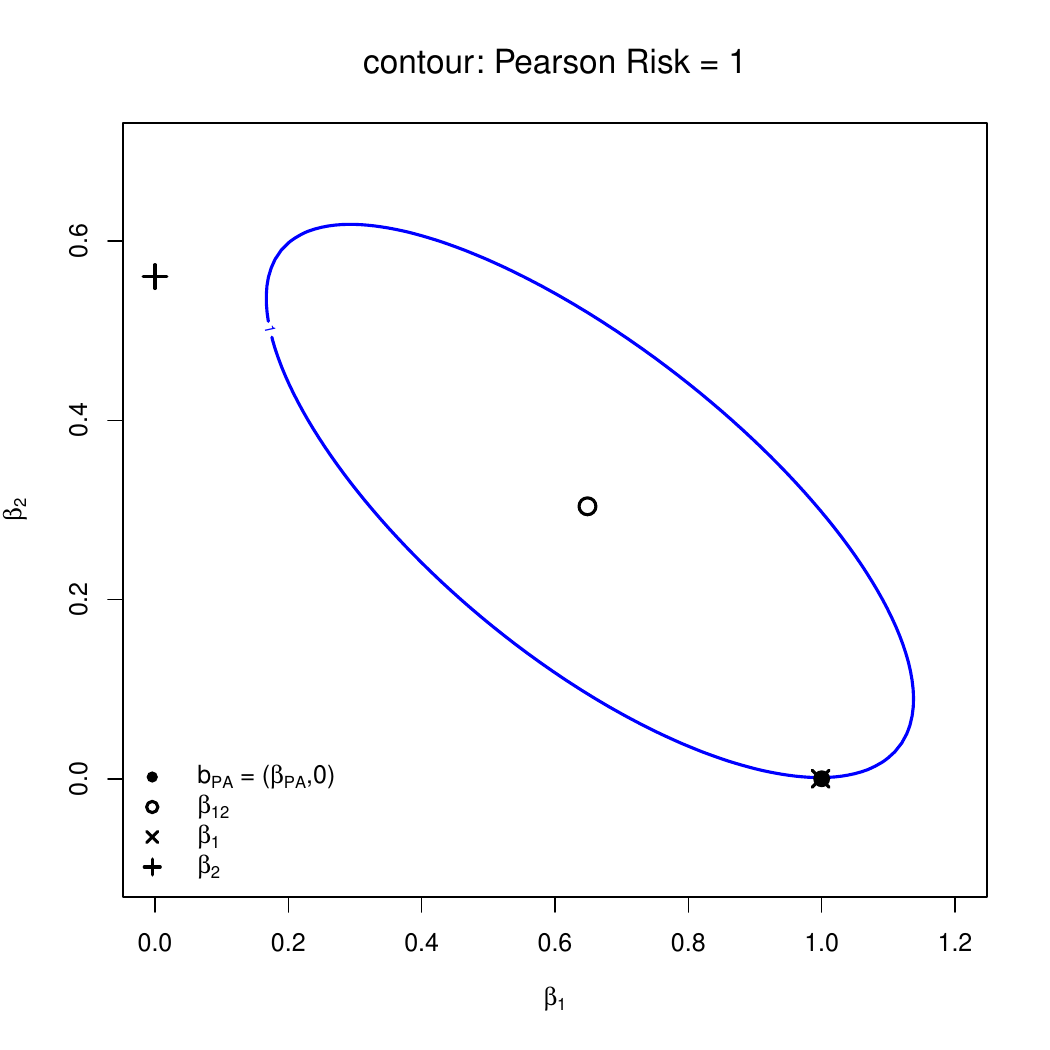}
	\caption{\label{contline_2cov}}
\end{subfigure}
\caption{Simulated data ($n=100000$) using a Poisson regression model with $p=2$ predictors. Contour lines of likelihood on (a) observational data and (b)-(c) two future environments, with covariates $X_1$ and $X_2$ perturbed by an additive noise with variance $\sigma^2=5$ and $\sigma^2=10$, respectively. The MLE solution (red dot) has a higher likelihood than the causal solution (blue dot) on the observational environment (a) but a lower likelihood on the shifted environments (b-c), particularly for larger noise (c). That is, the causal model has better out-of-sample guarantees. (d)  Contour line for Pearson risk $=1$ contains the true  causal parameters, which coincide with the ML $\bbeta_{1}$ of the model that includes only the causal parent $X_1$. The ML population estimates $\bbeta_{S}$ from other models are not on the Pearson risk $=1$ surface.}
\label{fig:contour-likelihood}
\end{figure}

The situation changes if one evaluates the prediction performance of the models on out-of-sample data. In particular, we consider the case where $X_1$ and $X_2$ are perturbed by an additional $N(0,\sigma^2)$ noise. Figures~\ref{fig:contour-likelihood-fut5} and \ref{fig:contour-likelihood-fut10} show the cases of $\sigma^2$ equal to $5$ and $10$, respectively. It is evident how the likelihood of the $\bbeta_{MLE}$ model evaluated on the out-of-sample data is now lower compared to that evaluated on the causal solution. This is more pronounced in Figure \ref{fig:contour-likelihood-fut10} than in Figure \ref{fig:contour-likelihood-fut5}, that is the larger the perturbation is. This shows, firstly, how causal models have better out-of-sample guarantees than predictive models, and, secondly, how maximising the likelihood on observational data is not an appropriate criterion for causal discovery. Both facts are well-known in the causal inference literature \citep[e.g.,][]{kania23}.

Figure~\ref{contline_2cov} shows the surface satisfied by all $\bbeta$ with Pearson risk equal to 1. As expected from Theorem~\ref{th:causal}, $\bb_{\mbox{\scriptsize PA}}=(1,0)$ belongs to this surface, and so does $\bbeta_{1}$, the population maximum likelihood estimator of the dependence of $Y$ only on the causal parent $X_1$. The remaining two models, i.e., the saturated model ($\beta_{12}$) and the model with $Y$ dependent only on $X_2$ ($\beta_{2}$), lead to population maximum likelihood estimates that do not lie on the Pearson risk surface. Indeed, these models are not causal and, therefore, do not enjoy the invariance property, as shown in the uniqueness result of Theorem~\ref{th:uniqueness}.

\subsection{Finite sample behaviour: causal Poisson regression}
\label{sec:simulation-finite}
In this section, we study the behaviour of the finite sample algorithm to obtain the causal parents of target $Y$.  We consider the more complex system in Figure~\ref{fig:sim7covariatesY}, where $Y$ has 2 causal parents ($X_2$ and $X_3$) and with 5 other variables in the system. 
\begin{figure}[tb!]
    \centering
    \begin{subfigure}{.39\linewidth}
    \centering 
		\begin{tikzpicture}
			\node[state] (X1)  {$X_1$};
			\node[state] (X2) [below= 0.5cm of X1 ]  {$X_2$};
			\node[state] (X3)[right = 1cm of X2] {$X_3$};	
			\node[state] (Y) [below= 2cm of X1,blue]  {$Y$};
                \node[state] (X4) [left= 1cm of Y]  {$X_4$};
                \node[state] (X5) [below= 2cm of X2]  {$X_5$};
                \node[state] (X6) [below= 2cm of X3]  {$X_6$};
                 \node[state] (X7) [below= 0.5cm of X6]  {$X_7$};
			\path[draw,thick,->]
			(X1) edge  (X2)
			(X1) edge  (X3)
			(X2) edge (Y)
                (X2) edge (X3)
                (X2) edge (X4)
                (X3) edge (Y)
                (Y) edge (X5)
	              (Y) edge (X6)
								(X2) edge[bend right] (X5)
                (X3) edge[bend left] (X5)
								(X2) edge[bend left] (X6)
                (X3) edge[bend left] (X6)
                (X6) edge (X7);
	\end{tikzpicture}
	\caption{}
\end{subfigure}
\begin{subfigure}{.59\linewidth}
\begin{equation*}
\begin{cases}
X_1&=\epsilon_1\\
X_2&=X_1+\epsilon_2\\
X_3&=X_1+X_2+\epsilon_3\\
Z&=\sin(5X_2)+ X_3^3+\epsilon_Z \\
Y&=F_{\rm{PA}}^{-1}(\Phi_{0,\sigma^2_Z}(\epsilon_{Z}))\\
X_4&=X_2+\epsilon_4\\
X_5&=Z+\epsilon_5\\
X_6&=Z+\epsilon_6\\
X_7&=X_6+\epsilon_7,
\end{cases}
\end{equation*}
\caption{\label{fig:sim7covariatesSEM}}
\end{subfigure}
\caption{Structural causal model for generating data on a target variable $Y$, $p=7$ predictors and the causal graph (a). The structural equations in (b), via a latent variable $Z$, show how the target $Y$ has a $\mbox{Poisson}(\sin(5X_2)+ X_3^3)$ distribution $F_{\rm{PA}}$ conditional on the causal parents $X_2$ and $X_3$.}
\label{fig:sim7covariatesY}
\end{figure}
The structural causal model in Figure \ref{fig:sim7covariatesSEM} shows how the distribution $F_{\rm{PA}}$ of $Y$ given its causal parents $X_2$ and $X_3$, is
\[ Y | \bX_{\mbox{\scriptsize PA}} = \bx_{\mbox{\scriptsize PA}} \sim \mbox{Poisson}\Big(e^{\sin(5x_2)+x_3^3}\Big),\]
that is, it is a Poisson regression where the predictor for the target $Y$ is described by the nonlinear  function $f_{\mbox{\scriptsize PA}}(x_2,x_3)=\sin(5x_2)+x_3^3$. The error terms are assumed to be independent and normally distributed with mean 0 and variance 0.04.  

\begin{figure}[h!]
\centering
\begin{subfigure}{.49\linewidth}
\centering
\includegraphics[scale=0.4]{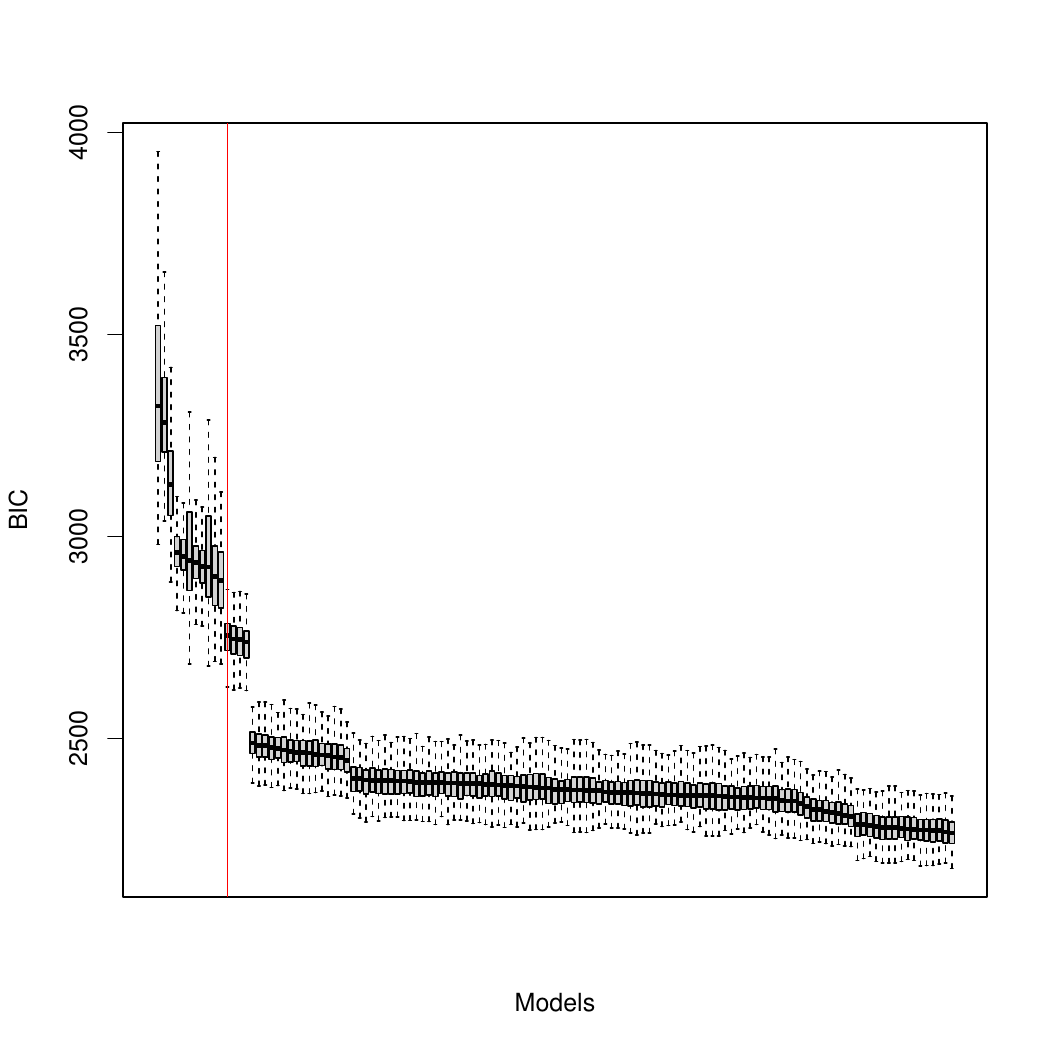}
\caption{\label{fig:BIC7}}
\end{subfigure}
\begin{subfigure}{.49\linewidth}
\centering
\includegraphics[scale=0.4]{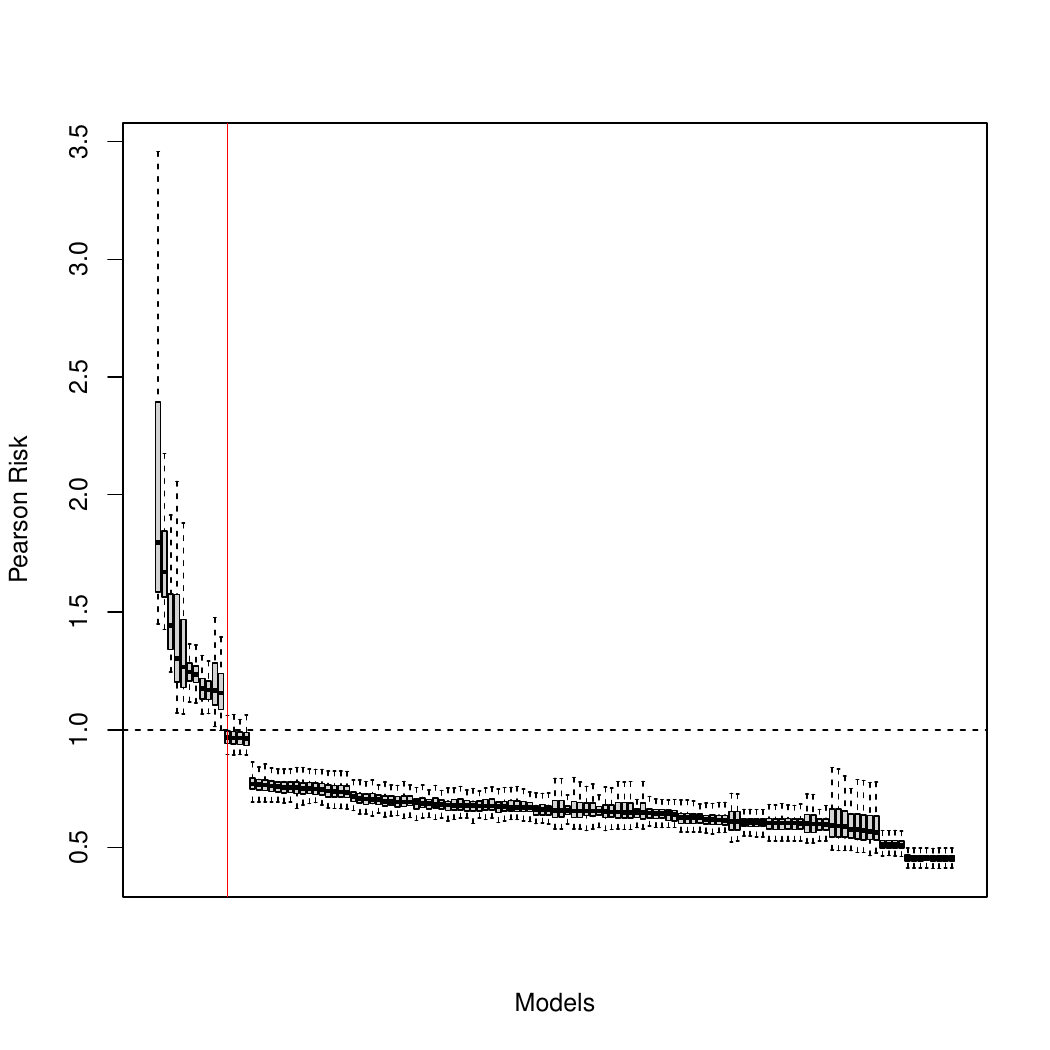}
\caption{\label{fig:Pearson7}}
\end{subfigure}
\begin{subfigure}{.49\linewidth}
\centering
\includegraphics[scale=0.4]{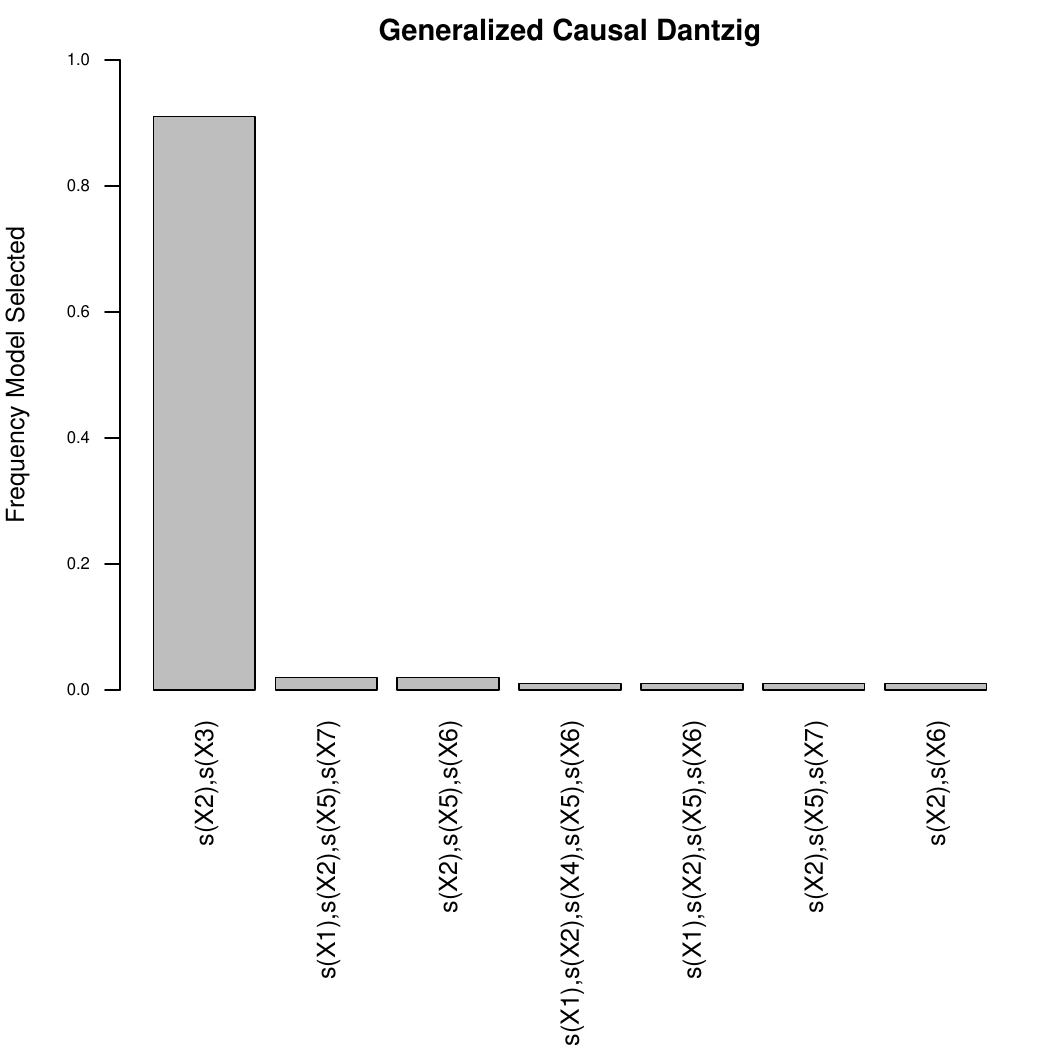}
\caption{\label{fig:optmod7}}
\end{subfigure}
\begin{subfigure}{.49\linewidth}
\begin{tabular}{|c|c|c|c|c|c|} \hline
 & \multicolumn{2}{c|}{Causal GLM} & \multicolumn{3}{c|}{PC-algorithm} \\ \cline{2-6}
$n$ & Full & Step & $X_2, X_3$ & $X_2$ & $X_3$\\ \hline
1000 & 91\% & 60\%  & 9\% & 24\% & 16\%\\
500 & 80\% & 64\% & 0\% & 14\% & 1\%\\
250 & 60\% & 49\% & 0\% & 11\% & 2\%\\
200 & 49\% & 41\% & 0\% & 9\% & 0\%\\
150 & 27\% & 33\% & 0\% & 2\% & 0\%\\
100 & 8\% & 8\% & 0\% & 1\% & 0\%\\ \hline
 \end{tabular}
\vskip 2.4cm
\caption{\label{tb:sim7covariates}}
\end{subfigure}
\caption{Results across 100 simulations with $n=1000$ from the Model in Figure \ref{fig:sim7covariatesY} with a Poisson distribution for $Y$ conditional on its causal parents. (a) BIC values across all models show how the causal model (red vertical line) is not the one that minimizes BIC; (b) Pearson risk values show how the causal model has a median Pearson risk close to 1; (c) \% times a model is selected by a causal GLM approach using a 5\% significance level shows a high percentage for the causal model; (d)The table reports the percentage of times that the correct model is selected by the proposed causal GLM approach, both under an exhaustive and a stepwise procedure. For the PC-algorithm, we count the graphs where there is a direct arrow to the target $Y$ from its parents ($X_2$, $X_3$), as well as those where only one of the two arrows is detected. Across all settings, the proposed method detects the true causal model more often than the PC-algorithm.}
\label{fig:sim7covariates}
\end{figure}

We repeat the simulation 100 times with sample size $n=1000$. We fit Poisson regression models for $Y$ for all possible subsets of the predictors. In order not to use any information from the generative process, we include also an intercept term in each of these models. Figure~\ref{fig:BIC7} shows the boxplots of BIC values of the models across the 100 simulations, sorted in descending order according to their median. The figure shows how the causal model (red vertical line) is far from being the most predictive model on the observational data. Indeed, the model with the lowest median BIC value is a model that includes the parents of $Y$ ($X_2$ and $X_3$) as well as its children ($X_5$ and $X_6$), which form the Markov blanket of $Y$ in the graph in Figure~\ref{fig:sim7covariatesY}. On the other hand, Figure~\ref{fig:Pearson7} shows how the causal model achieves a median Pearson risk close to 1, together with 3 other models that contain either ancestors ($X_1$) or variables blocked by the parents of $Y$ ($X_4$). Figure~\ref{fig:optmod7} shows how the causal GLM procedure is able to detect the correct causal model 91\% of the time out of all models, when setting an $\alpha=5\%$ in the asymptotic chi-squared statistical testing procedure. 

Figure~\ref{tb:sim7covariates} shows how the detection accuracy changes for decreasing sample sizes $n$. As expected, the accuracy in detecting the causal model decreases with $n$. Using a stepwise procedure rather than an exhaustive search generally leads to a reduction in the detection accuracy, but at a much lower computational cost (on average 5.2 times faster than the exhaustive search in this simulation study). In general, the accuracy achieved also by the stepwise procedure is significantly larger than what can be achieved with the PC-algorithm. As there are no implementations of the PC-algorithm that are specific to systems with Gaussian variables and a target count variable, we stabilized the variance of $Y$ by taking a square-root transformation and constructed the causal graph using Gaussian conditional independence tests on the resulting data.

\subsection{Finite sample behaviour: causal logistic regression}
\label{sec:simulation-logistic}
In this section, we evaluate the performance of the method in the case of a binary target variable $Y$, whose relationship with its causal parents is described by a logistic regression model.  We consider the system in Figure~\ref{fig:sim5covariates}, where $Y$ has 2 causal parents ($X_2$ and $X_3$) and with 3 other variables in the system.
\begin{figure}[tb!]
    \centering
    \begin{subfigure}{.39\linewidth}
    \centering 
		\begin{tikzpicture}
			\node[state] (X1)  {$X_1$};
			\node[state] (X2) [below= 0.5cm of X1 ]  {$X_2$};
			\node[state] (X3)[left = 1cm of X2] {$X_3$};	
			\node[state] (Y) [below= 1cm of $(X2)!0.5!(X3)$,blue]  {$Y$};
                \node[state] (X4) [below right= 0.8cm and 0.5cm of X2]  {$X_4$};
                \node[state] (X5) [below= 0.5cm of Y]  {$X_5$};
 		\path[draw,thick,->]
			(X1) edge  (X2)
			(X2) edge  (X4)
			(X2) edge (Y)
      (X3) edge (Y)
      (Y) edge (X5);
	\end{tikzpicture}
	\caption{}
\end{subfigure}
\begin{subfigure}{.59\linewidth}
\begin{equation*}
\begin{cases}
X_1&=\epsilon_1\\
X_2&=X_1+\epsilon_2\\
X_3&=\epsilon_3\\
Y& \sim \rm{Bernoulli}\Big(\dfrac{e^{0.8X_2-0.9X_3}}{1+e^{0.8X_2-0.9X_3}}\Big)\\
X_4&=X_2+\epsilon_4\\
X_5&=(1-\pi)Y+\pi (1-Y)+\epsilon_5,
\end{cases}
\end{equation*}
\caption{\label{fig:sim5covariatesSEM}}
\end{subfigure}
\caption{Structural causal model for generating data on a target variable $Y$, $p=5$ predictors and the causal graph (a). The structural equations in (b) show how the target $Y$ has a Bernoulli distribution conditional on the causal parents $X_2$ and $X_3$.}
\label{fig:sim5covariates}
\end{figure}

\begin{figure}[h!]
\centering
\begin{subfigure}{.49\linewidth}
\centering 
\includegraphics[scale=0.4]{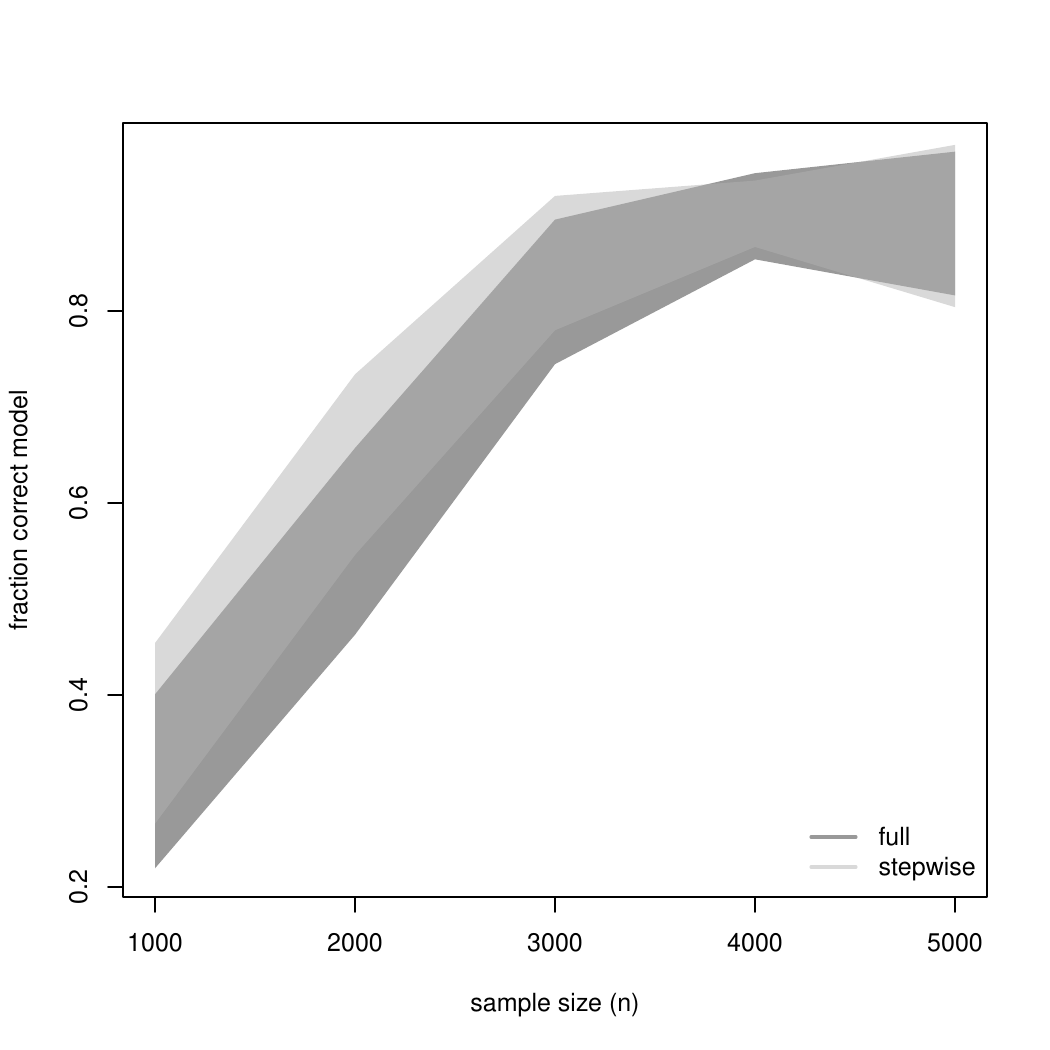}
\caption{\label{fig:logistic5cov-perc}}
\end{subfigure}
\begin{subfigure}{.49\linewidth}
\centering 
\includegraphics[scale=0.4]{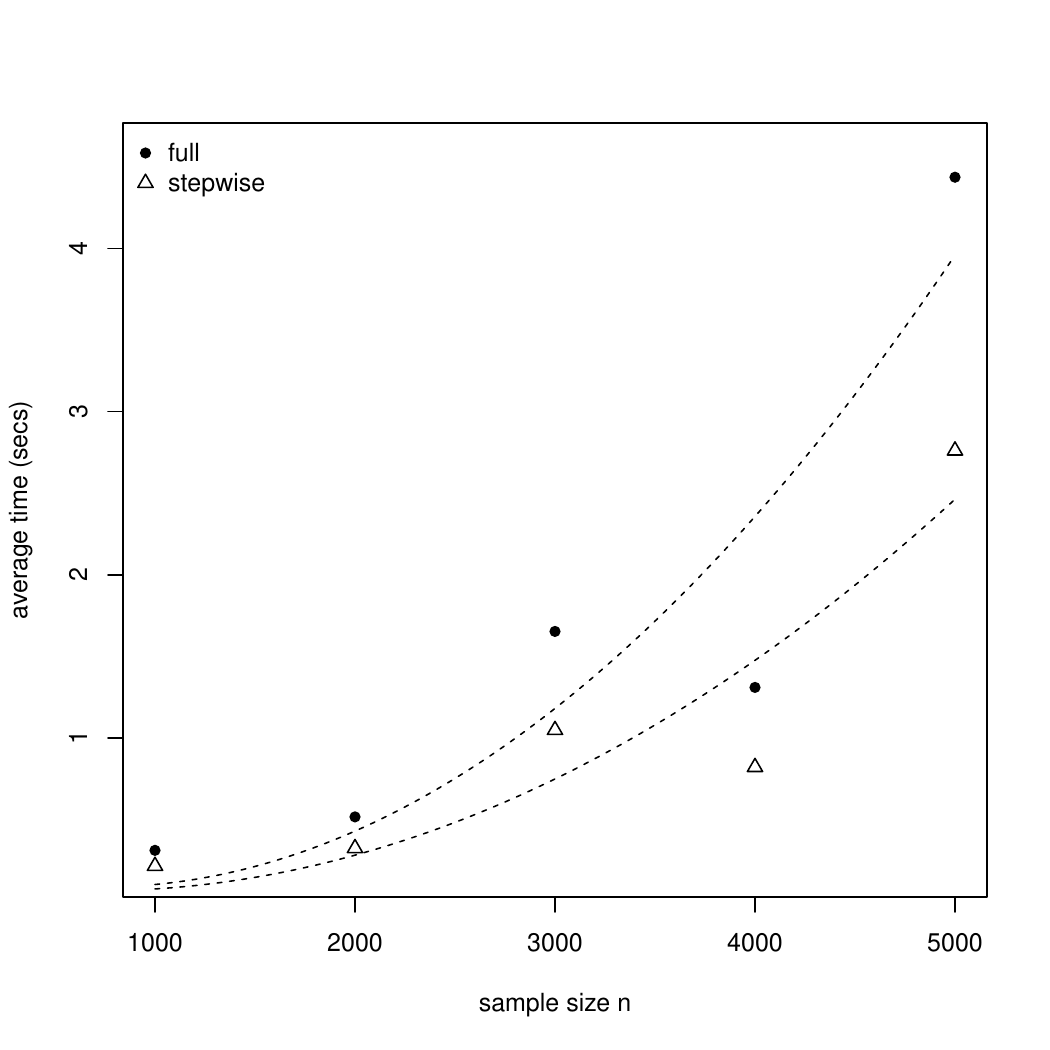}
\caption{\label{fig:logistic5cov-ct}}
\end{subfigure}
\caption{Results across 100 simulations with varying sample sizes from the Model in Figure \ref{fig:sim5covariates} with a Bernoulli distribution for $Y$ conditional on its causal parents. (a) The percentage of times that the correct model is selected by the proposed causal GLM approach, both under an exhaustive and a stepwise procedure, increases with $n$. (b) The stepwise procedure is computationally more efficient than the exhaustive search. 
\label{fig:logistic5cov}}
\end{figure}
We repeat the simulation 100 times with varying sample sizes and fit logistic regression models for $Y$ for all possible subsets of the predictors. For each model, a bootstrap test at a $5\%$ significance level is used to evaluate evidence against the null hypothesis of a perfectly dispersed Pearson risk.  The results in Figure \ref{fig:logistic5cov} show how, as with the previous simulation, the accuracy in detecting the true causal model increases with the sample size $n$. Moreover, particularly in the case when p-values are calculated via bootstrap, there is a significant gain in computational efficiency when using the stepwise procedure in place of the exhaustive search.

\section{Identifying causal drivers in empirical studies} \label{sec:realdata}

Here we illustrate the use of the causal GLM method on three empirical studies. In section~\ref{sec:chambers}, we consider data generated under a controlled experiment. In this case, there is knowledge about the true underlying physical model which can be used for validation. In section~\ref{sec:fertility} we consider data from a women's fertility study in the US. The aim is to identify the causal drivers of women's reproductive fertility. In section~\ref{sec:socioeco} we consider a socio-economic study for identifying potential causal drivers of high income. 

\subsection{Causal chambers: a controlled causal experiment}
\label{sec:chambers}
\cite{gamella25} describes a unique experimental program, called \emph{causal chambers}, in which physical laws are combined with a replicable experimental set-up to generate on-demand datasets that can be used for causal discovery validation. Here, we consider a dataset from a light tunnel set-up. The light tunnel contains a controllable light source and a variety of sensors, which allow to measure the light intensity at different points of the tunnel. 
Using the notation of  \cite{gamella25}, the following variables are of particular importance:
\begin{itemize}
\item $R$, $G$, $B$: Color setting of the light source
\item $\tilde{C}$: Current drawn by the light source
\item $\tilde{I}_1$,$\tilde{I}_2$, $\tilde{I}_3$: Infrared-intensity measurement at the first, second and third sensor
\item $\tilde{V}_1$, $\tilde{V}_2$, $\tilde{V}_3$: Visible-intensity measurement at the first, second and third sensor
\item $L_{i1}$, $L_{i2}$: Brightness of the two LEDs placed at the $i$th sensor
\item $\theta_1$, $\theta_2$: Position of the first and second polarizer.
\end{itemize}

\begin{figure}[h!]
\begin{subfigure}{.49\linewidth}
\begin{center}
\begin{tikzpicture}[
    scale=0.8,  
    node distance=1cm and 1cm,  
    every node/.style={draw=none, text centered},
    every path/.style={draw, -stealth}
]
\node (RGB) {\textbf{(R, G, B)}};
\node[left =0.5cm of RGB] (Ctilde) {$\tilde{C}$};
\node[below left=0.5cm and 1.5cm of RGB] (I1tilde) {$\tilde{I}_1$};
\node[below right=0.2cm and 0.1cm of I1tilde] (L12) {$L_{12}$};
\node[below left=0.5cm and 0.5cm of I1tilde] (L11) {$L_{11}$};
\node[below right =0.5cm and 0.8cm of L11] (V1tilde) {$\tilde{V}_1$};
\node[below right=0.5cm and 0.2cm of V1tilde] (I2tilde) {\textcolor{blue}{$\boldsymbol{\tilde{I}_2}$}};
\node[above right=0.2cm  and 0.2cm of I2tilde] (L21) {$\boldsymbol{L_{21}}$};
\node[below right=0.2cm  and 0.2cm of I2tilde] (L22) {$\boldsymbol{L_{22}}$};
\node[right=1cm of I2tilde] (V2tilde) {$\tilde{V}_2$};
\node[below=0.5cm of I2tilde] (D2I) {$\boldsymbol{D_2^I}$};
\node[below left=0.5cm and 0.5cm of I2tilde] (T2I) {$\boldsymbol{T_2^I}$};
\node[below right=0.05cm and 0.5cm of RGB] (Theta1) {$\theta_1$};
\node[right=0.2cm of Theta1] (Theta2) {$\theta_2$};
\node[above right=0.2cm and 0.2cm of V2tilde] (I3tilde) {$\tilde{I}_3$};
\node[right=0.4cm of I3tilde] (L31) {$L_{31}$};
\node[above right=0.2cm and 0.2cm of L31] (V3tilde) {$\tilde{V}_3$};
\node[below right=0.5cm and 0.1cm of I3tilde] (L32) {$L_{32}$};
\draw (RGB) -- (Ctilde);
\draw (RGB) -- (I1tilde);
\draw (RGB) -- (V1tilde);
\draw (L11) -- (I1tilde);
\draw (L12) -- (I1tilde);
\draw (L11) -- (V1tilde);
\draw (L12) -- (V1tilde);

\draw (RGB) -- (I2tilde);
\draw (RGB) -- (V2tilde);
\draw (L21) -- (I2tilde);
\draw (L21) -- (V2tilde);
\draw (L22) -- (I2tilde);
\draw (L22) -- (V2tilde);
\draw (D2I)--(I2tilde);
\draw (T2I)--(I2tilde);

\draw (RGB) -- (I3tilde);
\draw (RGB) -- (V3tilde);
\draw (L31) -- (I3tilde);
\draw (L31) -- (V3tilde);
\draw (L32) -- (I3tilde);
\draw (L32) -- (V3tilde);
\draw (Theta1) -- (I3tilde);
\draw (Theta2) -- (I3tilde);
\draw (Theta1) -- (V3tilde);
\draw (Theta2) -- (V3tilde);
\end{tikzpicture}
\end{center}
\caption{\label{fig:causalchamb-graph}}
\end{subfigure}
\begin{subfigure}{.49\linewidth}
\centering
\includegraphics[scale=0.4]{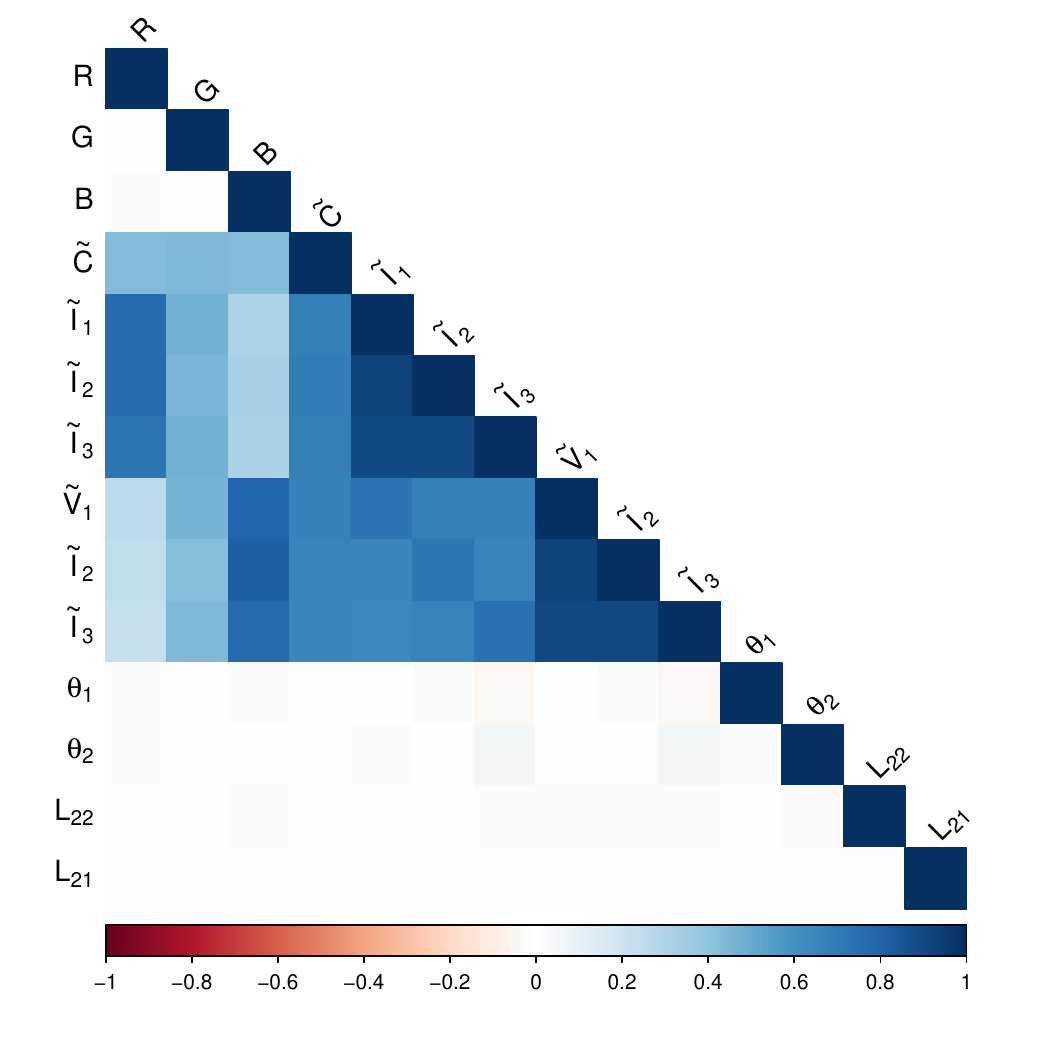}
\caption{\label{fig:causalchamb-corr}}
\end{subfigure}
\caption{(a) Causal graph underlying the light tunnel causal chamber experiment, taken from \cite{gamella25}. (b) Correlation heatmap among the main variables in the system.}
\label{fig:causalchambers}
\end{figure}
We focus on the variable $\tilde{I}_2$ in the uniform reference setting. A ground-truth causal graph of the system related to this variable is reported in Figure \ref{fig:causalchamb-graph} \citep{gamella25}.  Since in this environment the photodiodes used by the light sensors are set to the large infrared photodiodes ($D^I_2 = 2$) and their exposure time is set to the largest value ($T^I_2 = 3$), i.e., both are set to a constant, the causal parents of the target variable $\tilde{I}_2$ in the causal graph are the light source colors ($R,G,B$) and the brightness of the two LDEs placed at this second sensor ($L_{21}$, $L_{22}$).

As the variable $\tilde{I}_2$ is continuous, we construct a binary target variable  that indicates if the second light sensor $\tilde{I}_2$ is above a certain threshold, which we fix to $2000$. This means that $\tilde{I}_2^{\mbox{\scriptsize bin}}$ is Bernoulli distributed with a probability that depends on the distribution $F_{\tilde{I}_2}$ of $\tilde{I}_2$, which in turn depends on the causal parents, i.e.,
\begin{equation} \left\{ \begin{array}{l}
	\tilde{I}_2^{\mbox{\scriptsize bin}} \sim \mbox{Bern}(\pi) \\
	\pi = F^{-1}_{\tilde{I}_2}(2000; R,G,B, L_{21},L_{22})
\end{array}  \right.  \label{eq:chamber-truth}\end{equation}
The aim is to recover the causal parents of this target variable in the causal graph. As the response is binary, it is a natural choice to use the logistic additive regression version of the proposed causal approach, i.e., for $k=1,\ldots,10,000$,
\begin{equation} \left\{ \begin{array}{l}
	\tilde{I}_{2k}^{\mbox{\scriptsize bin}} \sim \mbox{Bern}(\pi_k) \\
	\mbox{logit}~ \pi_k = f_1(X_{k1})+ \ldots + f_p(X_{kp})
\end{array}  \right.  \label{eq:chamber-model}\end{equation}
where $X_1,\ldots,X_p$ represent the available variables in the uniform reference dataset provided by \cite{gamella25} as the observational environment for the analysis. However, it is important to note that the non-linear additive form in \eqref{eq:chamber-model} may not be rich enough to capture the true functional form in \eqref{eq:chamber-truth}. This may impede the discovery of the true causal parents. Moreover, as can be seen in Figure~\ref{fig:causalchamb-corr}, $L_{21}$ and $L_{22}$ appear to be not predictive of $\tilde{I}_{2}^{\mbox{\scriptsize bin}}$ and are not included in the Markov blanket under any analysis for the available number of observations $n$. For this reason, we exclude them from the causal discovery algorithm.

For the causal discovery, we restrict the search to the $p=9$ most predictive variables, found via a generalized additive model with all variables included and a null space penalization \citep{wood17}. In particular, we consider the following potential causal parents: $R$, $G$, $B$, $\tilde{C}$, $\tilde{I}_1$, $\tilde{I}_3$,  $\tilde{V}_1$, $\tilde{V}_2$, $\tilde{V}_3$. 
An exhaustive search across all $512$ models, using bootstrap for calculation of p-values and setting $\alpha=10\%$ as significance level, returns the causal model with $R$, $G$ and the three visible-intensity measurements as causal parents of $\tilde{I}_{2}^{\mbox{\scriptsize bin}}$. This is only partly reflecting the true generative process, since $B$ has not been detected and the non-causal variables $\tilde{V}_1$, $\tilde{V}_2$, and $\tilde{V}_3$ are most likely compensating (i) for the absence of this causal parent (indeed, they are the three most correlated variables with $B$ in Figure~\ref{fig:causalchamb-corr}) and (ii) for the miss-specification of the true link function in \eqref{eq:chamber-truth} by the additive logistic model in \eqref{eq:chamber-model}.

\subsection{Causal determinants of fertility} 
\label{sec:fertility} 
The National Opinion Resource Center's General Social Survey, which is currently available for 34 editions between 1972 and 2022 \citep{gss22}, can be used for the identification of the causal determinants of women's fertility. The survey is conducted on English-speaking women living in non-institutional arrangements in the United States. As in \cite{sander92}, we consider women between the age of $35$ and $54$, as younger women may not have completed their schooling or their reproductive period. Across the $34$ editions of the survey, this results in a sample size of $n= 16,269$ women. The target variable of interest is the \emph{number of children ever born}, whereas the potential causal determinants are investigated among $p=8$ available covariates:
\begin{itemize}
\item  level of education: years of schooling, mother's years of schooling, father's years of schooling (ordinal variables taking integer values between 0 and 20);
\item demographic variables: age, race (indicator variable for black), region of the country at age 16 (south, east, north-center, west), living environment at age 16 (farm, other rural, town, small city, large city);
\item time: the year of the study.
\end{itemize}

We fit an additive Poisson regression model with all variables, where we consider thin-plate regression splines for the 5 numerical variables (years of schooling, mother's and father's years of schooling, year of study, age) and factorial main effects for the remaining 3 categorical variables. The variables produce a relatively good fit, with 8.0\% of deviance explained. The results correspond to other association studies conducted in the literature, with years of schooling, age, race, living environment and years of study picked up as the most significant variables \citep{sander92,wooldridge09}. 

We now use the proposed approach to detect the causal determinants of fertility. A full search of all 256 models with $\alpha=5\%$ returns a model with years of schooling, mother's years of schooling, race, age, living environment and year of study as the causal determinants.  Figure~\ref{fig:fertility} shows the causal effects associated to the four numerical variables. 
\begin{figure}[h!]
\centering
\begin{subfigure}{.49\linewidth}
\centering
\includegraphics[scale=0.4]{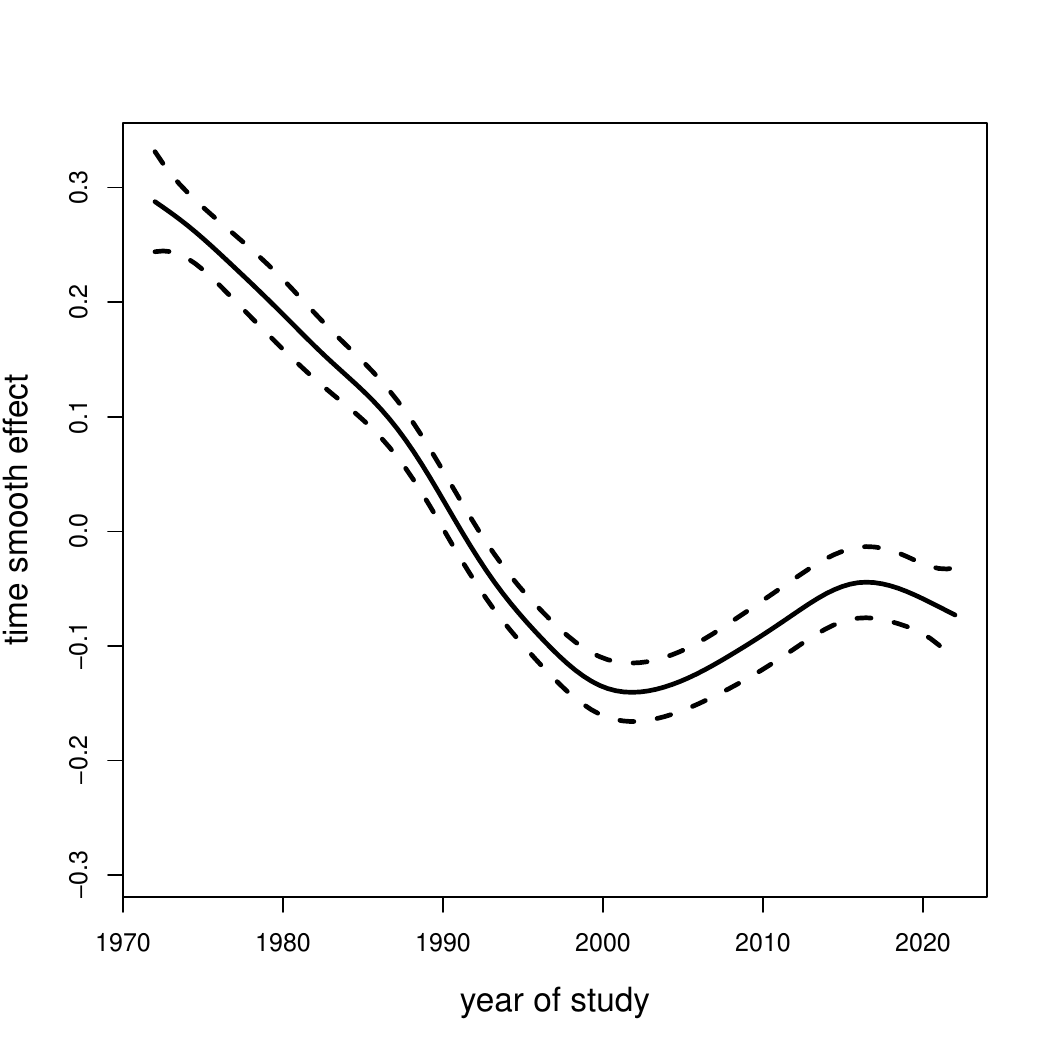}
\caption{\label{fig:fertilityyear}}
\end{subfigure}
\begin{subfigure}{.49\linewidth}
\centering
\includegraphics[scale=0.4]{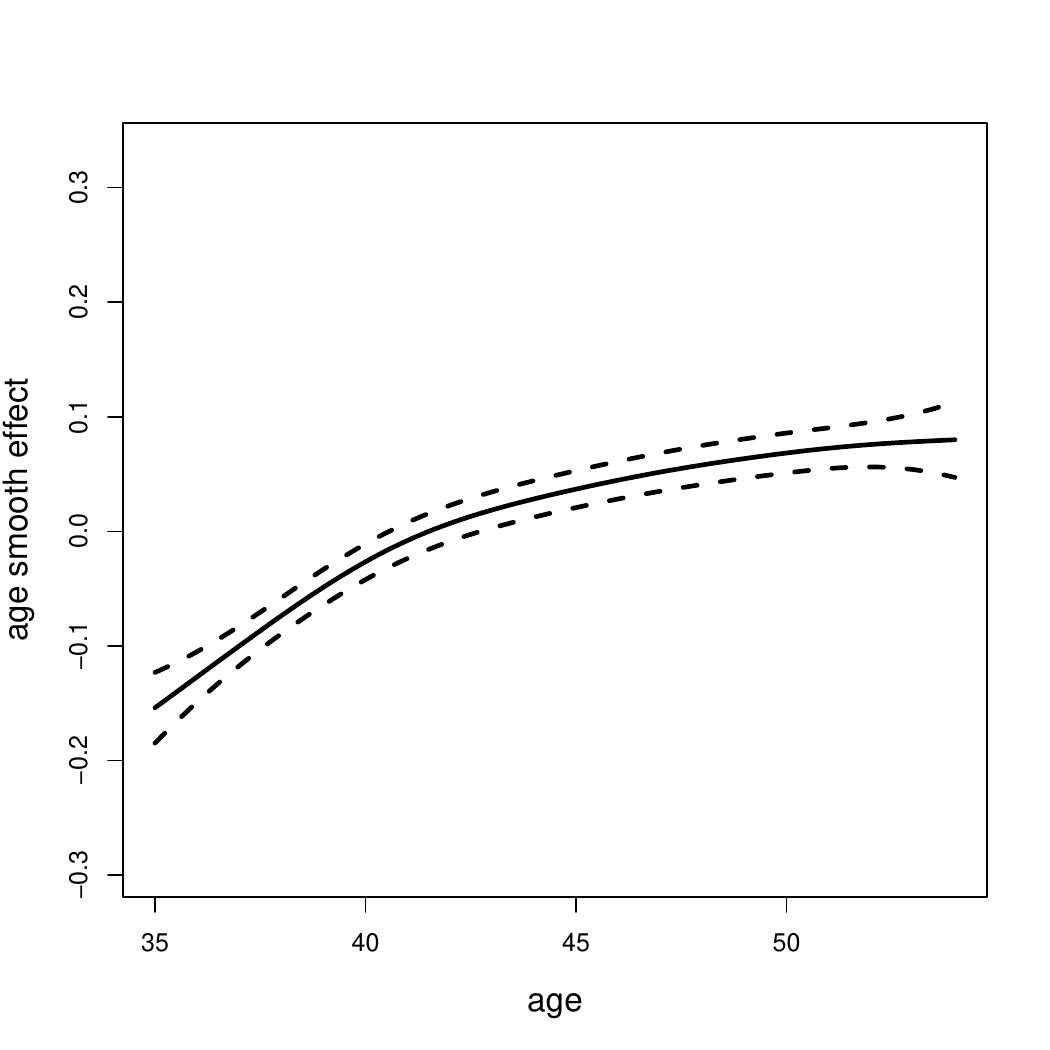}
\caption{\label{fig:fertilityage}}
\end{subfigure}
\begin{subfigure}{.49\linewidth}
\centering
\includegraphics[scale=0.4]{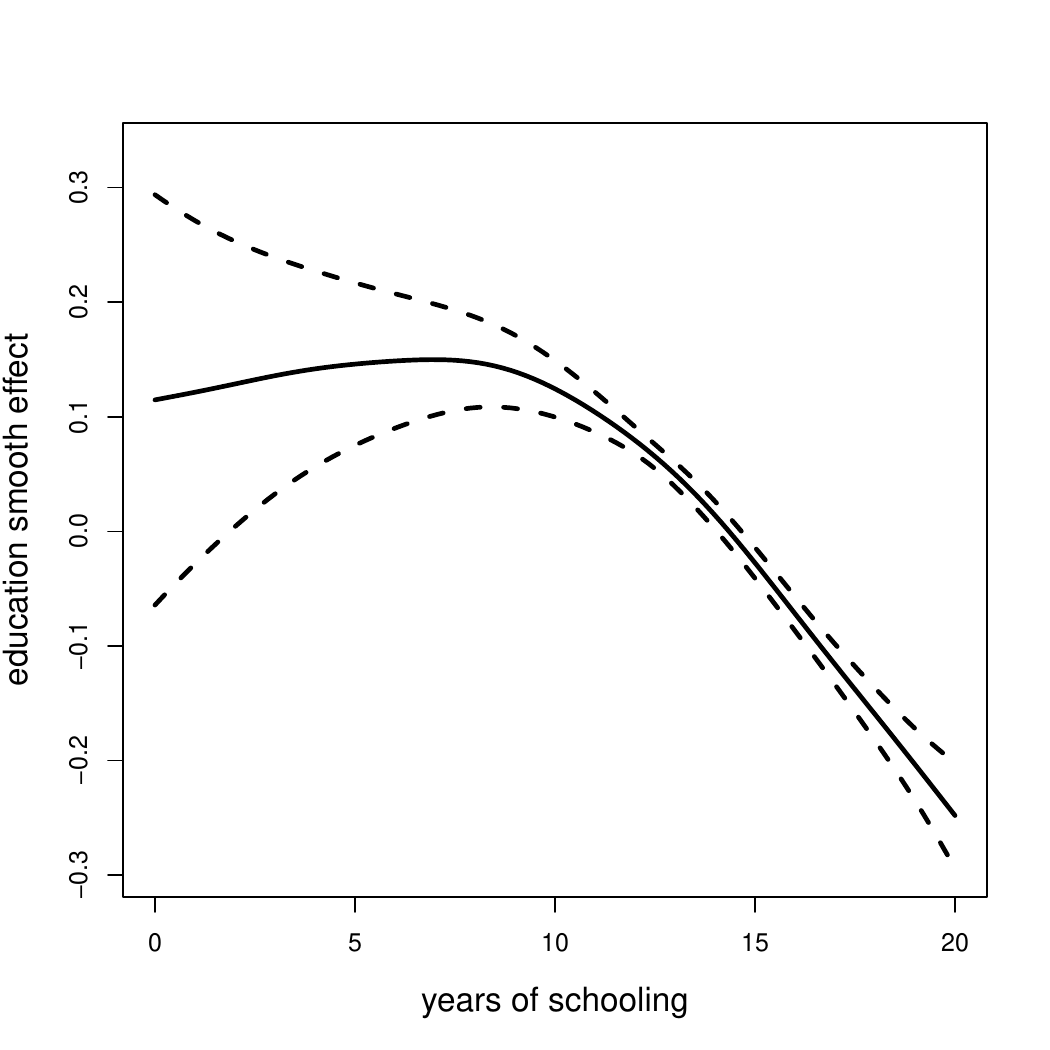}
\caption{\label{fig:fertilityeduc}}
\end{subfigure}
\begin{subfigure}{.49\linewidth}
\centering
\includegraphics[scale=0.4]{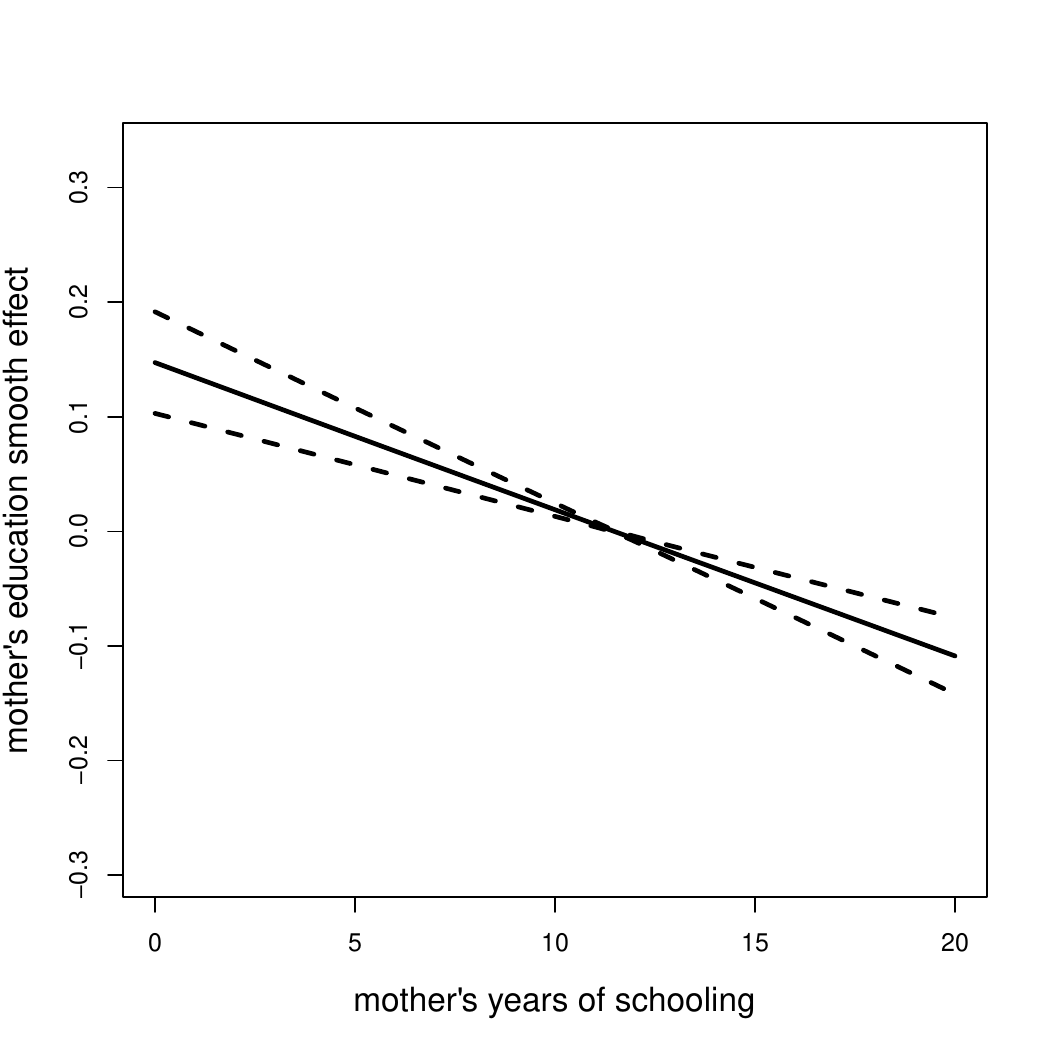}
\caption{\label{fig:fertilitymaeduc}}
\end{subfigure}
\caption{Causal determinants of women's fertility detected by the proposed invariant risk framework show how (a) time is associated to a decline in the number of children born; (b) age has a nonlinear increasing causal effect on the number of children born; (c) high levels of education cause a sharp drop in fertility; (d) mother's levels of education also cause a decrease in the number of children born.}
\label{fig:fertility}
\end{figure}
In particular, Figure~\ref{fig:fertilityyear} shows an overall decline over time in the number of children born, but with an interesting rise between 2000 and 2010. In the context of causality, this variable is most likely capturing unmeasured time-varying causal determinants of fertility. Figure~\ref{fig:fertilityage} shows a nonlinear and increasing causal effect of women's age on the number of children born, which is to be expected. The nonlinear association between this variable and the response is picked up also by \cite{wooldridge09} who includes a quadratic polynomial term for age in the linear model. It is interesting to note how a search of the causal model among the linear models does not pick up any causal determinant, even when one includes a quadratic term in age. This shows the flexibility of the proposed additive modelling framework in allowing to include smooth effects. Finally, Figure \ref{fig:fertilityeduc} shows a marked nonlinear causal effect of years of schooling on fertility, with high levels of education causing a sharp drop in fertility. This is pointed out also by \cite{sander92}. Although less pronounced, the same negative effect is found in relation to the mother's years of schooling. Among the categorical variables, we find a positive causal effect associated to being black, compared to other races (estimated causal coefficient 0.15) and to living in north-central or west American regions, compared to south (estimated causal coefficients both equal to 0.07). This is also in line with previous studies \citep{sander92}.

\subsection{Causal determinants of income}
\label{sec:socioeco}
The US census provides a wealth of socio-economic information. \cite{adult96} analyzed the 1994 census with the aim of developing a prediction model to find the factors that discriminate high earners ($> \$ 50K$ a year) from the rest ($ \le \$ 50K$ a year). The authors extracted a pre-processed subset of $n=32,561$ individuals from the original census data. While there are many association studies that have been conducted on this dataset, only few have considered it in the context of causality \citep{binkyte23}. We take this view and look for potential causal determinants amongst the following $p=8$ factors:
\begin{itemize}
\item level of education: ordinal variable taking integer values between 1 and 16;
\item demographic variables: age, race (white, asian-pacific-islander, american-indian-eskimo, other, black), gender, marital status (divorced, married, separated, single, widowed);
\item job-related variables: hours-per-week, occupation (blue-collar, white-collar, professional, sales, service, other/unknown), working class (government, private, self-employed, other/unknown). 
\end{itemize}

We apply the invariant risk method among the family of additive logistic regression models, where we consider thin-plate regression splines for the 3 numerical variables (education level, age and hours-per-week) and factorial main effects for the remaining 5 categorical variables. A full search of all 256 models with $\alpha=5\%$ returns a model with age, education level, marital status and occupation as the causal determinants.
\begin{figure}[h!]
\centering
\begin{subfigure}{.49\linewidth}
\centering
\includegraphics[scale=0.4]{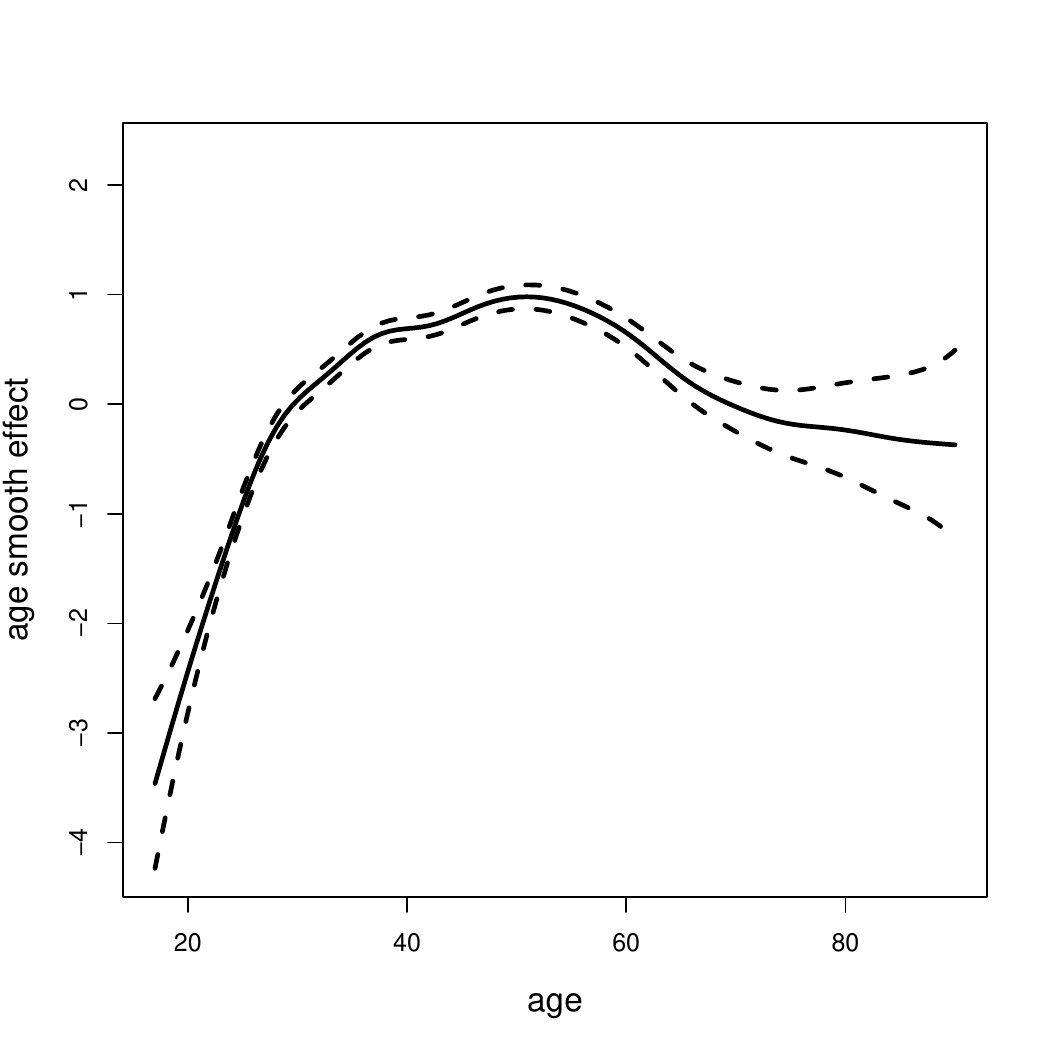}
\caption{\label{fig:incomeage}}
\end{subfigure}
\begin{subfigure}{.49\linewidth}
\centering
\includegraphics[scale=0.4]{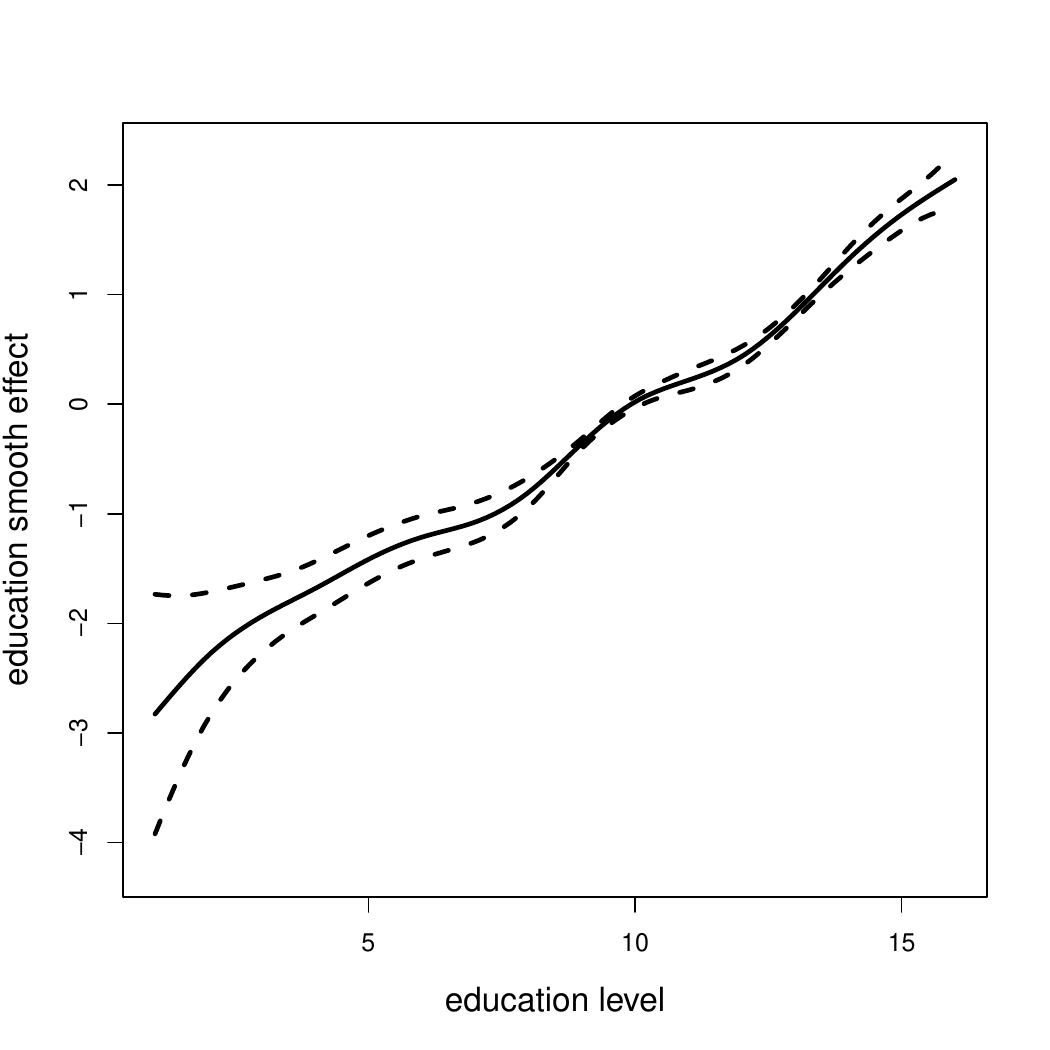}
\caption{\label{fig:incomeeduc}}
\end{subfigure}
\caption{Causal determinants of high income detected by the proposed invariant risk framework show how (a) age has an increasing causal effect on income, during the early working years; (b) education level has an increasing, close to linear, association with high earning.}
\label{fig:income}
\end{figure}

Figure \ref{fig:income}, for the two numerical variables, shows how age has a nonlinear causal effect on income, with a pronounced increasing effect during the early working years, while education level has an increasing, close to linear, association with high earning. As for the categorical variables, being male makes it 73\% more likely to be a high earner, whereas being married makes is roughly 7 times more likely than any other marital status level. With respect to type of occupation, white collar jobs, sales jobs and other professionals are between 60\% and 160\% more likely to be high earners, compared to blue collar jobs, services and other job types.

\section{Conclusion} \label{sec:conclusion}
In this paper, we have proposed a novel semi-parametric approach to causal discovery, when, conditional on its causal parents, the target variable of interest has a distribution in the exponential dispersion family. We have characterized the causal model by means of two properties: invariance of the expected Pearson risk and maximization of the expected likelihood. We have proposed a causal generalized linear model procedure, that consists of (i) a statistical testing procedure of the empirical Pearson risk, returning a list of candidate causal models; and (ii) a BIC selection step to refine the list of candidate models, as models containing variables that are d-separated from the target by the causal parents lead to the same perfectly dispersed Pearson risk as the causal model. 

For generalized linear models with a known dispersion parameter, such as for Poisson or Binomial, the Pearson risk is known. For these settings, the causal model can be recovered with data from a single environment. This is innovative in the context of invariant causal prediction, as the existing methods based on this idea require availability of data from multiple distinct experiments. Given the widespread use of Poisson and logistic regression models in association studies, this paper provides the means to tackle these studies from a causal perspective. For other distributional families, the method requires an a-priori value for the dispersion parameter, which may sometimes be available, such as for the negative binomial. If not available, data from multiple environments can be considered in order to fix the dispersion parameter across these environments, in line with similar invariance prediction methods available in the literature also for general response types \citep{kook23}.  

A further challenge of invariance prediction methods is the computational complexity of searching through the large space of potential causal models. In this paper, we propose a stepwise alternative to the exhaustive search across all models. The simulation studies show a natural reduction in computational cost without compromising accuracy too much. In order to increase computational efficiency, future work is looking at approximations of the test statistic under the null hypothesis, which are currently available only for some generalized linear models. In particular, we have used the fact that the Pearson statistics of a Poisson regression model can be well approximated by a chi-squared distribution under the null hypothesis to avoid time-consuming bootstrap procedures in this case.

\section*{Acknowledgments}
EW acknowledges support from the Swiss National Science Foundation (grant 192549). VV received financial support from the European Union (Next Generation EU Project PRIN 2022 PNRR, Mission 4 Component 2 Investment 1.1, Protocol Number P2022BNNEY, CUP E53D23016580001).

\bibliographystyle{chicago}
\bibliography{biblio}
\end{document}